\documentclass[12pt]{article}
\usepackage{amsmath}
\usepackage{amsfonts}
\usepackage{amssymb}
\usepackage{amsthm}
\usepackage{cite}
\usepackage{amsopn, amscd}
\usepackage[usenames,dvipsnames]{xcolor}
\usepackage[unicode=true,pdfusetitle,
 bookmarks=true,bookmarksnumbered=false,bookmarksopen=false,
 breaklinks=true,pdfborder={0 0 0},backref=false,colorlinks=true]
{hyperref}
\hypersetup{
 linkcolor=blue,citecolor=blue, urlcolor=blue}
\usepackage{mathrsfs}
\usepackage{subcaption}
\usepackage{cancel}
\usepackage{caption}
\newlength{\dinwidth}
\newlength{\dinmargin}

\newcommand{\bcc}{\color{black}}

\newcommand{\nhil}{\mathfrak{H}}
\newcommand{\HH}{P^0}

\newcommand{\pc}{\color{black}}

\newcommand{\gc}{\color{black}}
\newcommand{\bc}{\color{black}}

\DeclareMathOperator{\Stab}{Stab}

\DeclareMathOperator{\Ind}{Ind}

\newcommand{\EE}{\mrm{E}}

\def\Undertilde#1{\mathord{\vtop{\ialign{##\crcr
$\hfil\displaystyle{#1}\hfil$\crcr\noalign{\kern1.5pt\nointerlineskip}
$\hfil\widetilde{}\hfil$\crcr\noalign{\kern1.5pt}}}}}

\def\ti{\tilde}
\def\lan{\langle}
\def\ran{\rangle}

\def\cP{{\mathcal P}}

\def\tPoi{\widetilde{\cal P}_+^\uparrow}

\def\Ga{\Gamma}

\def\La{\Lambda}

\def\si{\sigma}

\def\om{\omega}
\def\Om{\Omega}

\def\setminus{\smallsetminus}

\def\eins{{\mathbf 1}}

\def\A{{\cal A}}
\def\B{{\cal B}}

\def\E{{\cal E}}
\def\F{{\cal F}}

\def\M{{\cal M}}
\def\N{{\cal N}}

\def\L{{\cal L}}
\def\cL{{\cal L}}

\def\O{{\cal O}}
\def\H{{\cal H}}
\def\K{{\cal K}}
\def\S{{\cal S}}
\def\V{{\cal V}}
\def\cU{{\cal U}}

\def\ad{{\operatorname{ad}}}

\def\S2{S^{1(2)}}
\def\Poi{{\cal P_+^\uparrow}}

\def\setminus{\smallsetminus}

\def\Z{{\mathbb Z}}

\def\RR{{\mathbb R}}
\def\CC{{\mathbb C}}
\def\NN{{\mathbb N}}
\def\ZZ{{\mathbb Z}}

\def\sl2{{{\rm SL}(2,\RR)}}
\def\psl2{{{\rm PSL}(2,\RR)}}

\def\u1{{{\rm V}(1)}}
\def\su2{{{\rm SV}(2)}}

\def\so3{{{\rm SO}(3)}}

\def\W{{\mathcal W}}

\def\SO{{\mathrm{SO}}}
\def\ind{{\mathrm{Ind}}}
\def\SLC{{\mathrm{SL}(2,\CC)}}

\usepackage[numbers]{natbib}
\usepackage{amsfonts}
\usepackage{textcomp}
\usepackage{url}
\usepackage{amsbsy}
\usepackage[utf8]{inputenc}
\usepackage{bbold}
\usepackage{accents}

\usepackage{tikz}
\usetikzlibrary{trees}
\usetikzlibrary{decorations.pathmorphing}
\usetikzlibrary{decorations.markings}

\newcommand{\y}{\!\!\!}

\newcommand{\beps}{\bar\eps}
\newcommand{\mrm}{\mathrm}
\renewcommand{\mathbf}{\boldsymbol}

\newcommand{\Span}{\mathrm{Span}}

\newcommand{\bv}{\mathbf{v}}

\newcommand{\wtla}{\la}
\newcommand{\wtU}{U}

\newcommand{\tout}{\overset{\out}{\times}}

\newcommand{\bn}{\mathbf{n}}

\newcommand{\bP}{\mathbf{P}}
\newcommand{\loc}{\mathrm{loc}}

\newcommand{\bx}{\mathbf{x}}

\newcommand{\mcF}{\mathcal A}
\newcommand{\mcP}{\mathcal P}
\newcommand{\mcL}{\mathcal L}
\newcommand{\mcK}{\mathcal K}

\newcommand{\out}{\mathrm{out}}

\renewcommand{\i}{\mathrm i}

\newcommand{\wt}{\widetilde}

\newcommand{\ga}{\gamma}

\newcommand{\be}{\beta}

\newcommand{\pa}{\partial}

\newcommand{\ov}{\overline}

\newcommand{\eps}{\varepsilon}

\newcommand{\De}{\Delta}

\newcommand{\nin}{\noindent}
\newcommand{\ph}{\phantom}

\newcommand{\nat}{\mathbb{N}}
\newcommand{\hil}{\mathcal{H}}

\newcommand{\mco}{\mathcal{O}}

\newcommand{\supp}{\mathrm{supp}}
\newcommand{\fr}[2]{\frac{#1}{#2}}
\newcommand{\al}{\alpha}
\newcommand{\real}{\mathbb{R}}
\newcommand{\complex}{\mathbb{C}}
\newcommand{\la}{\lambda}
\newcommand{\non}{\nonumber}

\newtheorem{theoreme}{Theorem } [section]
\newtheorem{proposition}[theoreme]{Proposition}
\newtheorem{lemma}[theoreme]{Lemma}
\newtheorem{definition}[theoreme]{Definition}
\newtheorem{corollary}[theoreme]{Corollary}
\newtheorem{remark}[theoreme]{Remark}
\newtheorem{example}[theoreme]{Example}
\newtheorem{criterion}[theoreme]{Criterion}
\newtheorem{conjecture}{Conjecture}
\newtheorem{assumption}{Assumption}

\newcommand{\bea}{\begin{assumption}}
	\newcommand{\eea}{\end{assumption}}
\newcommand{\beco}{\begin{conjecture} }
	\newcommand{\eeco}{\end{conjecture} }
\newcommand{\beq}{\begin{equation}}
	\newcommand{\eeq}{\end{equation}}
\newcommand{\beqa}{\begin{eqnarray}}
	\newcommand{\eeqa}{\end{eqnarray}}
\newcommand{\ben}{\begin{arabicenumerate}}
	\newcommand{\een}{\end{arabicenumerate}}
\newcommand{\bex}{\begin{example}}
	\newcommand{\eex}{\end{example}}
\newcommand{\ber}{\begin{remark}}
	\newcommand{\eer}{\end{remark}}
\newcommand{\bec}{\begin{corollary}}
	\newcommand{\eec}{\end{corollary}}
\newcommand{\bep}{\begin{proposition}}
	\newcommand{\eep}{\end{proposition}}
\newcommand{\becr}{\begin{criterion}}
	\newcommand{\eecr}{\end{criterion}}

\begin{document}
\title{ {\bc The} Bisognano-Wichmann property for asymptotically complete massless  QFT} 

\author{
{\bf Wojciech Dybalski}\\
Zentrum Mathematik, Technische Universit\"at M\"unchen,\\
E-mail: {\tt dybalski@ma.tum.de}
\and
{\bf Vincenzo Morinelli}\\
Dipartimento di Matematica, Universit\`a di Roma ``Tor
Vergata''\\ 
E-mail: {\tt morinell@mat.uniroma2.it}}

\date{}

\maketitle

\begin{abstract} We prove the Bisognano-Wichmann  property for asymptotically complete Haag-Kastler 
theories of massless particles.  These particles should either be scalar or appear as a direct sum of two opposite integer 
helicities, thus, e.g., photons are covered. The argument relies on a \emph{modularity condition} formulated recently by
one of us (VM) and on the Buchholz' scattering theory of massless particles.

\end{abstract}

\section{Introduction}
\setcounter{equation}{0}

For any von Neumann algebra with a faithful state Tomita-Takesaki theory gives a natural dynamics constructed using the complex structure of the algebra.  
It was shown by Bisognano and Wichmann that for the algebra of the Wightman fields localised in a spacelike wedge  this latter dynamics coincides with the Lorentz boosts in the direction of this wedge \cite{BW76}. Furthermore, in  the presence of the Bisognano-Wichmann (B-W) property  the full global symmetry of the model is contained in the modular structure of the net reducing the dichotomy between symmetries and algebras to an inclusion \cite{BGL95}.  The B-W property is also important for many other reasons,
ranging from the intrinsic meaning of the CPT symmetry \cite{GL} to a  construction of interacting models \cite{Le08}  and to entanglement theory \cite{Witt}. While its formulation is most natural in the algebraic
(Haag-Kastler) setting, and it is known to hold in all the `physical' examples, its general proof in this framework is missing to date. The reason is the broadness of the Haag-Kastler setting which admits also non-physical counterexamples to the B-W property.   {\bc For} example, when an infinite family of massive spinorial or infinite spin particles occurs  \cite{LMR16,Mor}. Thus it is important to find natural assumptions which exclude such pathological cases.

For massless theories the assumption of  global conformal invariance implies the B-W property as shown in \cite{BGL93}. A search for an algebraic sufficient condition for the B-W property, not relying on conformal covariance, was started in \cite{Mor, Mo16} at the level of one particle nets. Here a criterion on the covariant representation called the 
\emph{modularity condition} was shown to give  the B-W property of the one particle net.  For massive theories which are asymptotically complete, 
the paper by Mund \cite{Mu01} gives the B-W property. This paper exploits a result of Buchholz and Epstein \cite{BE85}  in order to study geometrically the analytic extension of one parameter boosts and identify it with the  associated modular operator. This allows to verify the Bisognano-Wichmann property on the one particle subspace and conclude it for the full interacting net by asymptotic completeness and wedge localization of {\bc the} modular operator. Unfortunately this method  does not apply to massless theories as the argument of Buchholz and Epstein requires a mass gap.

In the present paper we prove the B-W property for massless bosonic theories which are asymptotically complete,  by combining  some ideas contained in the works mentioned above. First, we identify  the modular operator and the boost generator associated to the same wedge  at the single-particle level. To this end, we verify the 
modularity  condition introduced in  \cite{Mo16, Mor}  and  thus avoid the use of the Buchholz-Epstein result. 
 Our argument requires that the representation of the Poincar\'e group is either scalar or a direct sum of two 
 representations with opposite integer helicities.
 Thereby  we show that  the modularity condition applies to a large family of massless representations, including higher  helicity. 
Next, we show that the B-W property holds on the entire Hilbert space by using scattering theory and  the assumption of asymptotic completeness. We recall that scattering theory for massless  bosons was developed in \cite{Bu77} and various
simplifications  have been found meanwhile. In the present paper we use the variant from \cite{AD17} which is based on novel ergodic theorem arguments and on uniform energy bounds on asymptotic fields from \cite{Bu90, He14}.
 Our results 
extend the range of validity of the B-W property and reconfirm its status as a generic property of physically reasonable models.  

Our paper is organized as follows: In   Sect.~\ref{2} we state our main result after  the necessary preparations.
In Sect.~\ref{Preliminaries-section} we recall some relevant facts from scattering theory of massless particles, the theory of standard subspaces
and one-particle nets, {\bc and representations of the Poincar\'e group}. In Sect.~\ref{modularity-sect}  the modularity condition is stated and verified for one-particle massless  nets with arbitrary integer spin.
In Sect.~\ref{last-section} the result is generalised to an arbitrary number of particles using scattering theory and the assumption of asymptotic completeness. 

\vspace{0.2cm}

\noindent{\bf Acknowledgment:}  W.D. would like to thank Sabina Alazzawi who was involved at early stages of this project.
 Both  authors thank  Maximilian Duell for interesting discussions. 
W.D. was supported by the {\bc Deutsche Forschungsgemeinschaft} (DFG) within the Emmy Noether grants DY107/2-1 {\bcc and DY107/2-2}.
V.M.  Titolare di un Assegno di Ricerca
dell'Istituto Nazionale di Alta Matematica (INdAM fellowship), supported in part by the ERC Advanced Grant 669240 QUEST ``Quantum Algebraic Structures and Models'', MIUR FARE R16X5RB55W QUEST-NET, GNAMPA-INdAM,  acknowledge  the MIUR Excellence Department Project awarded to the Department of Mathematics, University of Rome Tor Vergata, CUP E83C18000100006.
\vspace{0.2cm}

\section{Framework and results}\label{2}
\setcounter{equation}{0}
\subsection{Local nets and the Bisognano-Wichmann property}\label{local-nets}
Let $\real^{1+3}$ be the  Minkowski spacetime.
We denote by $\mcK$ the family
of double cones $\mco\subset   \real^{1+3}$ ordered by inclusion and write $\mco'$ for the causal complement of $\mco$ in 
$\real^{1+3}$. Furthermore, let $\wt\mcP_+^\uparrow=\real^4\rtimes \SLC$ denote the covering group of the proper ortochronous Poincar\'e group 
$\mcP_+^\uparrow$.  We {\bc denote} with $\Lambda:\widetilde\mcP_+^\uparrow\rightarrow \mcP_+^\uparrow$ the covering map.

\begin{definition}\label{HK}  {\pc Let $\nhil$ be a fixed Hilbert space.} {\bcc W}e say that $\mcK\ni\mco\mapsto \mcF(\mco)\subset B(\nhil)$ is a local net of von Neumann algebras in a vacuum representation if
the following properties hold:
\begin{enumerate}
\item \textbf{Isotony:} $\mcF(\mco_1)\subset \mcF(\mco_2)$ for $\mco_{1}\subset \mco_2$.

\item \textbf{Poincar\'e covariance:} there is a continuous unitary representation $\wtU$ of  $\wt\mcP_+^\uparrow$  such that
\beqa
\wtU(\wtla)\mcF(\mco)\wtU(\wtla)^*=\mcF(\la \mco)\quad  \mathrm{for}\quad 
\wtla\in \wt\mcP_+^\uparrow.
\eeqa
\item \textbf{Positivity of the energy:} the joint spectrum of translations in $U$ is contained in the forward lightcone $V_+=\{ p\in\RR^{1+3}:  p^0\geq 0, \, p^2=(p,p)\geq 0\}$.

\item \textbf{Cyclicity of the vacuum:} there is a unique (up to a phase) unit vector $\Om\in\nhil$, the physical vacuum state, which is $U$-invariant and cyclic for the global algebra {\pc $\mcF:=\overline{\bigcup_{\O\subset \RR^{1+3}}\A(\O)}^{\|\,\cdot\,\|}$ of the net.}

\item \textbf{Locality:} $\, \mcF(\mco_{1})\subseteq \mcF(\mco_{2})'$  for $\mco_{1}\subset \mco_{2}'$.

\end{enumerate}
A local net of von Neumann algebras will be denoted by $(\mcF,U, \Om)$.
\end{definition}
{\pc For future reference, we set for any region $\mathcal U\subset \real^{1+3}$
\beqa
\mcF_{\loc}(\mathcal{U}):=\bigcup_{\mco\subset \mathcal U} \mcF(\mco) 
\quad \textrm{ and } \quad \mcF(\mathcal{U}):=\mcF_{\loc}(\mathcal U)''\label{algebra-for-large-regions}
\eeqa
and we refer to $\mcF_{\loc}:=\mcF_{\loc}(\real^{1+3})$ as the algebra of strictly local operators. 
}

 In order to introduce the B-W property, we need some geometric preliminaries: 
a {\bf wedge shaped region} $W\subset\RR^{1+3}$ is an open region of the form  $gW_1$ where $g\in\Poi$ and $W_1=\{x\in\RR^{1+3}: |x_0|<x_1\}$.    The set of wedges is  denoted by $\W$.  {\pc  Note that if  $W\in\W$,  then $W'\in\W$,  
 where prime denotes here the spacelike complement. }
It is possible to associate to any wedge a one-parameter group of boosts  $\La_W$ fixing the wedge $W$ by  the following
formula   for $W_1$ 
\begin{equation}\label{boost}
 \mathbb{R}\ni t\rightarrow\Lambda_{W_{1}}(t):=\begin{pmatrix}
\cosh(t)&\sinh( t)& 0& 0\\
\sinh( t)&\cosh(t)& 0& 0\\
0& 0& 1& 0\\
0& 0& 0& 1
\end{pmatrix}
\end{equation}
and the covariant action of the Poincar\'e group on the set of wedges.  

{\pc 
 We  call 
$W_\alpha=\{x\in\RR^{1+3}: |x_0|<x_\alpha\}$, $\alpha=1,2,3$, the wedge in the $x_\alpha$ direction  and with $R_\alpha$, $\Lambda_\alpha$ are the one-parameter groups  of rotations and boosts, respectively, 
fixing $W_\alpha$.  Their unique one parameter group lifts to $\SLC$ are  denoted $r_\alpha$ and $\lambda_\alpha$. In general $\lambda_W$ will denote the one parameter group lift of $\Lambda_W$. Note that $\lambda_\alpha(t)=e^{\frac t2\,\sigma_\alpha }$ and $r_\alpha(\theta)=e^{i\frac\theta 2 \sigma_\alpha }$ where $t,\theta\in\RR$ and $\sigma_\alpha$ are the Pauli matrices. In particular one has that $r_\alpha(2\pi)=-I=: r(2\pi)$.}

For any $W\in \W$ we define $\mcF(W)$ according to (\ref{algebra-for-large-regions}).
It is well known that the vacuum is cyclic and separating for $\mcF(W)$ thus the Tomita-Takesaki theory gives the corresponding modular
evolution $\real\ni t\mapsto \De_{W}^{it}$.

\begin{definition} 
We say that a local net $(\mcF,U, \Om)$ satisfies the \textbf{ Bisognano-Wichmann property}  if for all  $W\in\W$, $t\in\real$, 
\beqa
 U(\lambda_W(2\pi t))=\Delta^{-it}_{W}. \non
\eeqa
\end{definition}

\subsection{Massless Wigner particles and asymptotic nets} \label{massless-section}

 Scattering theory of massless Wigner particles was  developed by Buchholz \cite{Bu77, Bu75}, both in the bosonic and fermionic case. Recently   the bosonic case was  simplified in \cite{AD17}.  We collect below the main results in this subject  following \cite{AD17}.
We first introduce the single-particle subspace.
\begin{definition}\label{single-particle-definition} A local net $(\mcF,U, \Om)$ describes massless Wigner particles 
 if $\nhil$ contains a subspace $\nhil^{(1)}\neq \{0\}$ s.t.
 \beqa
\mrm{Ran}\mathbf{1}_{\{0\}}( M)=\complex \Om \oplus \nhil^{(1)}, 
 \eeqa 
where $\mathbf{1}_{\{0\}}(M)$ denotes the spectral projection of the mass operator  $M:=\sqrt{(P^0)^2-\pmb{P}^2}$ corresponding
to the  eigenvalue zero. 
 We say that these particles have helicities $h_1,h_2,h_3\ldots\in \mathbb{Z}$ if $U  |_{\nhil^{(1)}}$ is {\bc a finite or infinite multiple of} the direct sum of the
corresponding zero mass  representations.   
\end{definition}

It is well known that to any local theory containing massless particles one can associate an asymptotic (free) theory \cite{Bu77}.
We outline now this construction following \cite{AD17}. 
For the unitary representation of translations $\wtU |_{\real^4}$ we shall write
$\wtU(x)=e^{i(\HH x^0-\bP \cdot \bx)}$ and
for translates of  observables $A\in \mcF$ the notations $\al_x(A):=A(x):=\wtU(x)A\wtU(x)^*$ are used. If $g\in L^1(\real^4)$, then $A(g):=\int A(x) g(x) d^4x$
{\bcc denotes the operator $A$  smeared with the function $g$.}
Moreover, we set 
\beqa
\mcF_{\loc,0}:=\{A\in \mcF_{\loc}  \,:\,   x\mapsto A(x)\, \textrm{ smooth in norm}\}.
\eeqa
{\bcc This is a weakly dense $*$-subalgebra of $\mcF_{\loc}$, as can be seen by smearing local operators with delta-approximating functions.}
{\bcc Next,} we specify the following 
 Poincar\'e invariant subset of $C_0^{\infty}(\real^4)$
\beq
C_*(\real^4):=\{ (n_{\mu}\pa^{\mu})^{5}g    \,:\,    g\in C_0^{\infty}(\real^4), \,\, n_0=\sqrt{1+\bn^2}\,\} \label{S-star-def}
\eeq
{\bcc and} define 
\begin{align}
\mcF_{C_*}&:=\{\, B(g)   \,:\,    B\in \mcF_{\loc,0}, \ g\in C_*(\real^4) \,\}, \\
\mcF^{C_*}&:=\Span\, \mcF_{C_*}, \\
\mcF_{C_*}(\mco)&:=\mcF_{C_*}\cap \mcF(\mco), \quad \mcF^{C_*}(\mco):=\mcF^{C_*}\cap \mcF(\mco), 
\quad  \mco\in \mcK. \label{local-C}
\end{align}

Now we move on to the construction of asymptotic fields of massless particles. For any $A\in \mcF^{C_*}$ and $f\in C^{\infty}(S^2)$,  we set   as in \cite{Bu77,Bu82}
\beqa
A_t\{f\}:=-2\,t\int d\om(\bn)\,f(\mathbf{n})\,\partial_0A(t,t\mathbf{n}).
\eeqa
Here $d\om(\mathbf{n})=\frac{\sin\nu\,d\nu d\varphi}{4\pi}$ is the normalized, invariant measure on $S^2$ and 
$\pa_0A:=\partial_s(\mathrm{e}^{\i sH}A\mathrm{e}^{-\i sH})|_{s=0}$. In order to improve the convergence in the limit of large $t$, we proceed to  time averages of $A_t\{f\}$, namely
\begin{equation}\label{timeAverage}
\bar{A}_t\{f\}:=\int\,dt'\,h_t(t')\,A_{t'}\{f\}.
\end{equation}
Here for non-negative $h\in C_0^{\infty}(\real)$, supported in the interval $[-1,1]$ and normalized so that
$\int dt\, h(t)=1$, we set $h_t(t')=t^{-{\beps}} h(t^{-{\beps}}(t'-t))$ with
$t\geq 1$ and $0<\beps<1$.  It turns out that these limits exist on all vectors from the domain
\beqa
D_{P^0}:=\bigcap_{n\geq 1} D((P^{0})^n),
\eeqa
where $D((P^{0})^n)$ is the domain of self-adjointness of $(P^{0})^n$.
\begin{lemma}\label{asymptotic-fields}
Let $A\in\mcF^{C_*}(\mco)$ and $f\in C^{\infty}(S^2)$. Then, the limit
\begin{equation}
A^{\out}\{f\}\Psi=\lim_{t\to\infty}\bar{A}_t\{f\}\Psi \label{first-A-out}
\end{equation}
exists for $\Psi\in D_{P^0}$  and is again an element of $D_{P^0}$. 
\end{lemma}
The operators $A^{\out}\{f\}$ are constructed in such a way that they create single-particle states from the
vacuum, namely
\beqa
 A^{\out}\{f\}\Om=P^{(1)} f\left(\tfrac{\mathbf{P}}{|\mathbf{P}|}\right)A\Om,  \label{single-particle-photon}
 \eeqa
where $P^{(1)}$ is the projection on the single-particle subspace $\nhil^{(1)}$. {Vectors of the form (\ref{single-particle-photon})
span a dense subspace of $\nhil^{(1)}$, even in the case $f\equiv 1 $}. Furthermore, if $A^{\out}\{f\}$, $A'^{\out}\{f'\}$ 
are two asymptotic fields as specified above, then
\beqa
\,[A^{\out}\{f\}, A'^{\out}\{f'\}]=\lan \Om,[A^{\out}\{f\}, A'^{\out}\{f'\}]\Om\ran 1_{\nhil} \label{c-number-commutator}
\eeqa
as operators on $D_{P^0}$.  For $f\equiv 1$ the operators $A^{\out}\{f\}$ appearing in Lemma~\ref{asymptotic-fields} are denoted $A^{\out}$ and are 
called the asymptotic fields.
For $A=A^*$ these operators are essentially self-adjoint on $D(P^0)$ and their self-adjoint extensions are
denoted by the same symbol.  For any $\mco\in \mcK$ we introduce the von Neumann algebra:
\beqa
\mcF^{\out}(\mco):=\{\, e^{i A^{\out}} \,:\, A\in \mcF^{C_*}(\mco), \  A^*=A \,\}''.\label{asymptotic-algebra}
\eeqa
The triple $(\mcF^{\mrm{out}},U, \Om)$ satisfies all the properties  from Definition~\ref{HK}, 
except, perhaps, for the cyclicity of the vacuum. If the latter property also holds, then we say that the theory   $(\mcF,U, \Om)$ is \textbf{asymptotically complete}. 
{\bcc Clearly, for the definition of asymptotic completeness the case $f\equiv 1$ suffices. However, the operators $A^{\out}\{ f\}$ for other choices of $f$ will
be needed in Sect.~\ref{last-section} at the technical level. For this reason we  collected their properties above.}

Now we are ready to state the main result of this paper:
\begin{theoreme} \label{main-theorem} Let $(\mcF, U,\Om)$ be a local net containing massless particles with helicity zero or  with 
helicities $(h, -h)$ for some $h\in \nat$. If this net is
 asymptotically complete, then it satisfies the Bisognano-Wichmann property.
\end{theoreme}
\begin{proof} Follows from Theorem~\ref{one-particle-net-BW} and Propositions~\ref{last-section-proposition}, \ref{final-proposition-scattering} below. \end{proof}
Even if the original net $(\mcF, U,\Om)$ is not asymptotically complete, we can set $\nhil^{\out}:={\bc \ov{\mcF^{\out}\Om}}$ 
and define the asymptotic net $(\mcF^{\out} |_{ \nhil^{\out} }, U |_{ \nhil^{\out} }, \Om)$ which is asymptotically complete
by construction. In view of the  commutation relations (\ref{c-number-commutator}), this net can be considered free, but it
is not automatically the net of the corresponding textbook free field theory\footnote{For example, if we choose as the original net the
`truncated' net, s.t. the local algebras of regions below certain size are declared to be $\complex 1$, the  asymptotic net will inherit this property  {\bc \cite{AD17}}.}.
By Theorem~\ref{main-theorem}, this net satisfies the Bisognano-Wichmann property if it contains
massless  Wigner particles with helicity zero or $(h, -h)$, $h\in \nat$.

{\bcc We note that the local nets satisfying the assumptions of Theorem~\ref{main-theorem} are a posteriori in the setting of \cite{GL}. Indeed,
the modular covariance is an obvious consequence of the Bisognano-Wichmann property and the Reeh-Schlieder property for spacelike cones
follows from the spectrum condition and cyclicity of the vacuum under $\mcF$ (cf. \cite[Appendix]{Bu75}). From the spin-statistics theorem of \cite{GL}
it follows that such nets are actually covariant under $\Poi$. Furthermore, by the CPT theorems of this reference, the unitary representation $U$ of $\Poi$
extends to an (anti-)unitary covariant representation of $\cP_+$ (the group generated by $\Poi$ and the PT operator $\Theta$)  as follows:}
\begin{corollary} Let $(\mcF, U,\Om)$ be a local net as in Theorem \ref{main-theorem}. Then $U$ extends to an (anti-)unitary representation of the Poincar\'e group $\cP_+$ by
$$J_{W_1}U(R_1(\pi))=U(\Theta)$$
where $\Theta {\bcc x}=-{\bcc x}$ with ${\bcc x}\in\RR^{1+3}$, $J_{W_1}$ is the modular conjugation associated to $(\mcF(W_1),\Omega)$ and $R_1(\pi)$ is the 
{\bcc rotation by $\pi$ around the first axis.}
\end{corollary}

\section{Preliminaries} \label{Preliminaries-section}
\setcounter{equation}{0}

\subsection{Scattering states of massless particles}

In this subsection we provide some preparatory information about the  Hilbert space of scattering states $\nhil^{\out}:={\bc \ov{\mcF^{\out}\Om}}$
introduced above.
Namely, we extract the creation and annihilation parts of the asymptotic fields (\ref{first-A-out}) in order to facilitate the construction of scattering states.
We still follow \cite{AD17} which  in turn relied here on \cite{DH15}.  
Let $\theta\in C^{\infty}(\real)$, $0\leq \theta\leq 1$, be supported in $(0,\infty)$ and equal to one on $(1,\infty)$. Moreover, let $\be\in C_0^{\infty}(\real^4)$, $0\leq \be\leq 1$, be equal to one in some neighbourhood of zero and satisfy $\be(-p)=\be(p)$. Furthermore, for a parameter $1\leq r<\infty$ and a future oriented timelike unit vector $n$ we define
\beqa
\wt\eta_{\pm,r}(p):=\theta(\pm r (n_{\mu}p^{\mu}))\be(r^{-1} p), \label{test-functions-zero}
\eeqa 
where tilde denotes the Fourier transform. As $r\to \infty$ these functions approximate the characteristic functions of the positive/negative energy half planes $\{\, p\in \real^4 \,:\,  \pm n_{\mu}p^{\mu}\geq 0\,\}$. We also have $\bar \eta_{\pm,r}=\eta_{\mp,r}$. Note that the family of functions $\eta_{\pm,r}$, as specified above, is invariant under Lorentz transformations.
\bep\label{creation-annihilation} \emph{\cite{Bu77,AD17}} Let $A\in \mcF_{C_*}$, $f\in C^{\infty}(S^2)$. Suppose that 
the timelike unit vectors $n$ entering the definition of $A$ and of $\eta_{\pm,r}$ coincide. Then:
\begin{enumerate}
\item[(a)] The limits 
$A^{\out}\{f\}^{\pm}\Psi:=\lim_{r\to \infty}A^{\out}\{f\}(\eta_{\pm,r})\Psi$, $\Psi\in D_{P^0}$, 
exist and define the creation and annihilation parts of $A^{\out}\{f\}$ as  operators on $D_{P^0}$. $A^{\out}\{f\}^{\pm}$
do not depend on the choice of the functions $\theta$ and $\be$ in (\ref{test-functions-zero}) within the specified restrictions.

\item[(b)]  $(A^{\out}\{f\}^{\pm})^*|_{D_{P^{0}} }=A^{*\out}\{\bar f\}^{\mp}$. In particular,  $A^{\out}\{f\}^{\pm}$ are closable operators.

\item[(c)] $A^{\out}\{f\}^{\pm}D_{P^0}\subset D_{P^0}$.

\item[(d)] $A^{\out}\{f\}=A^{\out}\{f\}^{+}+A^{\out}\{f\}^{-}$ on $D_{P^0}$.
\end{enumerate}
\eep
Making use of Proposition~\ref{creation-annihilation} and of (\ref{c-number-commutator}) we also obtain on $D_{P^0}$ 
\beqa
\,[A^{\out}\{f\}^-, A'^{\out}\{f'\}^+  ]=\lan A^{*\out}\{\bar{f}\}^+\Om, A'^{\out}\{f'\}^+ \Om\ran  1_{\nhil}  \label{CCR-asymptotic}
\eeqa
and the commutators of pairs of creation (resp. annihilation) operators vanish. The following definition of scattering states is slightly more general than in \cite{Bu77,AD17}, as we do not assume $f \equiv 1$. The proof is an obvious application of the canonical commutation relations~(\ref{CCR-asymptotic}). 
\bep\label{photon-scattering}  \emph{\cite{Bu77,AD17}} The states $\Psi^{\out}:= A_1^{\out}\{f_1\}^+ \ldots A_n^{\out}\{ f_n\}^+\Om$ have the following properties:
\begin{enumerate}
\item[(a)] $\Psi^{\out}$ depends only on the single-particle states $\Phi_i=A_i^{\out}\{f\}\Om\in \nhil^{(1)}$. Therefore, we write 
$\Psi^{\out}=\Phi_1\tout\cdots \tout\Phi_n$.
\item[(b)] 
$\lan  \Phi_1\tout\cdots \tout\Phi_n,\Phi'_1\tout\cdots \tout\Phi'_{n'}\ran=
\delta_{n,n'}\sum_{\si\in \mathfrak{S}_n}\lan \Phi_1, \Phi'_{\si_1}\ran\ldots \lan \Phi_n, \Phi'_{\si_n}\ran$, where
$\mathfrak{S}_n$ is the set of all permutations of $(1,\ldots, n)$.
\end{enumerate}
\eep
\nin  The subspace of $\nhil$ spanned by vectors of the form $\Psi^{\out}=\Phi_1\tout\cdots \tout\Phi_n$ for fixed $n$ will be denoted $\nhil^{(n)}$.
We note that 
\beqa
\nhil^{\out}:={\bc {\ov{\mcF^{\out} \Om}}}=\bigoplus_{n\geq 0}\nhil^{(n)},  \label{H-out}
\eeqa
where $\nhil^{(0)}=\complex \Omega$ and $\nhil^{(1)}$ was introduced in Definition~\ref{single-particle-definition}. 
The last equality in (\ref{H-out}) follows from density of vectors of the form (\ref{single-particle-photon}) 
in $\nhil^{(1)}$ and from the canonical commutation relations (\ref{CCR-asymptotic}) by standard Fock space arguments.
Clearly,  $\nhil^{\out}$ is naturally isomorphic to the symmetric Fock space over $\nhil^{(1)}$, denoted $\Ga(\nhil^{(1)})$.

\subsection{Standard subspaces}

 We recall here some elements of the theory of standard subspaces following \cite{L}. {\pc In the later part of this subsection we also provide several  results which we were not able to find in literature and that will be needed in our investigation.}

A  real linear, closed subspace $H$ of a complex Hilbert space $\H$ is called {\bf cyclic} if 
$H+iH$ is dense in $\H$, {\bf separating} if $H\cap iH=\{0\}$ and 
{\bf standard} if it is cyclic and separating.

Given a standard subspace $H$  the associated {\bf Tomita operator} $S_H$ is defined to be the closed anti-linear involution with domain $H+iH$, given by: $$S_H:H+iH\ni \xi + i\eta \mapsto \xi - i\eta\in H+iH, \qquad\xi,\eta\in H.$$ The polar decomposition $$S_H = J_H\Delta_H^{1/2}$$ defines the positive 
self-adjoint {\bf modular operator} $\Delta_H$ and the anti-unitary
{\bf modular conjugation} $J_H$. $\Delta_H$ is invertible and $J_H\Delta_H J_H=\Delta_H^{-1}.$

Let $H$ be a  real 
linear subspace of $\H$, the \emph{symplectic complement} of $H$ is defined by
\[
H' := \{\xi\in\H\ :\ {\bc \mrm{Im}\lan \xi,\eta\ran}=0, \forall \eta\in H\} = (iH)^{\bot_\RR}\ ,
\]
where $\bot_\RR$ denotes the orthogonal {\bc complement} in $\H$  with respect to the real part of the scalar product on $\H$.
$H'$ is a closed, real linear subspace of $\H$. 
It is a fact that $H$ is cyclic (resp. separating) iff $H'$ is separating (resp. cyclic), thus $H$ is standard iff $H'$ is standard and in this case
\[
S_{H'} = S^*_H \ ,
\]
with $J_{H'}=J_{H}$ and $\De_{H'}=\De^{-1}_{H}=J_{H}\De_{H} J_{H}$ \cite{L}.
 Furthermore, if $H$ is standard, then $H = H''$.
We recall that the one-parameter, strongly continuous group $t\mapsto \Delta_H^{it}$  is called the {\bf modular group} of $H$ and
\begin{equation*}
\Delta_H^{it}H = H, \quad J_H H = H' \ ,\qquad  t\in\RR\ .
\end{equation*}

 There is a 1-1 correspondence between Tomita operators and  standard subspaces. 
 \begin{proposition}\label{prop:11}{\rm \cite{L}}.
The map 
\begin{equation}\label{SH}H\longmapsto S_H
\end{equation}
is a bijection between the set of standard subspaces of $\H$ and the set of closed, densely defined, anti-linear involutions on $\H$. 
\end{proposition}
 The following  are  three  basic  results on  standard subspaces.
\begin{lemma}\label{lem:sym}{\rm \pc \cite{Mor}.}
Let $H,K\subset\H$ be standard subspaces  and $U\in\cU(\H)$ be a unitary operator on $\H$ such that $UH=K$. Then $U\Delta_H U^*=\Delta_K$ and $UJ_HU^*=J_K$.
\end{lemma}

\begin{lemma}\label{inc}{\rm \cite{L}.}
Let $H\subset \H$ be a standard subspace, and $K\subset H$ be  a closed, real linear subspace of $H$. 
If $\Delta_H^{it}K=K$, $\forall t\in\RR$, then $K$ is a standard subspace of $\K:= \overline{K+iK}$ and $\Delta_H |_K$  is the modular operator of $K$ on $\K$. 
Moreover, if $K$ is a cyclic subspace of $\H$, then $H=K$.
\end{lemma}
\begin{theoreme}{\rm \cite{L}.} \label{Borch}
Let $H\subset\H$ be a standard subspace, and $U(t)=e^{itP}$ be a one-parameter unitary group on $\H$ with {\bc a} generator $\pm P>0$, such that $U(t)H\subset H$, $\forall t\geq 0$. Then, \begin{equation}\label{eq:bor}\left\{\begin{array}{ll}
\Delta_H^{is}U(t)\Delta_H^{-is}= U(e^{\mp2\pi s}t)&\\J_{H}U(t)J_{H}=U(-t)\end{array}\right.\qquad \forall t,s\in\RR.\end{equation}
\end{theoreme}
 \nin We note that the above  result is a variant of the Borchers theorem \cite{Bo92, Flo} for standard subspaces.

 The following three lemmas, {\bcc which} {\pc we could not find in {\bcc the} literature}, will be needed to analyze the subspaces $H^{(1)}(W)$ defined in (\ref{H-one}) below.   
They concern decompositions of standard subspaces w.r.t.  projections $E$ commuting with $S_H$.
Since $S_H$ is unbounded and not self-adjoint, we mean here that $E$ commutes with $J_{H}$ and bounded Borel functions of $\De_{H}$.  
If $\xi'\in D(S_{H}^*)=H'+iH'$ and $\xi\in D(S_{H})=H+iH$, then, for such $E$
\begin{align}
\lan \xi', S_{H}E\xi\ran&=\ov{\lan S_{H}^*\xi', E\xi\ran}=  \ov{\lan J_{H}\De^{-1/2}_{H}\xi', E\xi\ran}=  \lim_{n\to\infty} \lan \chi_n(\De^{-1/2}_{H}) {\bcc \De^{-1/2}_{H}} \xi', E J_{H}  \xi\ran\non\\
&= \lim_{n\to\infty} \lan \xi', E \chi_n(\De^{-1/2}_{H}){\bcc \De^{-1/2}_{H}} J_{H}  \xi\ran=  \lan \xi', E \De^{-1/2}_{H} J_{H}  \xi\ran= \lan \xi', E S_{H}  \xi\ran, \label{detailed-commutation}
\end{align}
where $\chi_n$ is the characteristic function of $[-n,n]$ and we made use of the fact that $\xi', J_{H}\xi\in D( \De^{-1/2}_{H})$ to control the 
limit $n\to\infty$.
\begin{lemma}\label{lem:ort1}
Let  $H\subset\H $ be a standard subspace and $E=E^2=E^*$ be a
projection  commuting with $S_H$. Then $H=EH \oplus (1-E)H$. Furthermore,  $EH$ and $(1-E)H$ are standard in 
$E\H$ and $(1-E)\H$, respectively.
\end{lemma}
\begin{proof}
$H$ is defined to be the kernel of $1-S_H$. Now since $E$  commutes with $S_H$,
for every $\xi\in H$, $E\xi \in \mathrm{Ker} (1-S_H)$, thus $E\xi \in H$   (cf. computation (\ref{detailed-commutation}) above). 
{\bc As the same argument applies to $(1-E)$} and  $\xi=E\xi+(1-E)\xi$,  we have the claim. The last statement is obvious. \end{proof}
\vspace{-0.7cm}
\begin{lemma}\label{lem:ort2}
Let  $H\subset\H $ be a standard subspace and $E=E^2=E^*$ be a projection  commuting with $S_H$. Then $H'= EH'\oplus  (1-E)H'$. Furthermore, $EH'$ and $(1-E)H'$ are standard in $E\H$ and $(1-E)\H$, respectively.
\end{lemma}
\begin{proof}
If $E\in\B(\H)$ is a projection commuting with $S_H$, then $E$ also commutes
with $S_{H'}=J_H\Delta_H^{-1/2}=\Delta^{1/2}_H {J_H}$ and
the decomposition $H'= EH'\oplus  (1-E)H'$ follows as in Lemma \ref{lem:ort1}. {\bc We note that} $S_H$ and $S_{H'}$ and their polar decompositions decompose through $E$. {\bc (Clearly, $S_{EH}=S_{H} |_{E\hil}$ and $S_{(1-E)H}=S_{H} |_{(1-E)\hil}$)}. {\bc Consequently}, $(EH)'=EH'$ and $((1-E)H)'=(1-E)H'$  on $E\H$ and $(1-E)\H$,  respectively. {\bc Since} $EH$ and $(1-E)H$ are cyclic and separating in $E\H$ and $(1-E)\H$, {\bc the claim follows}. 
\end{proof} 
An immediate consequence is:
\begin{lemma}\label{lem:cons}
Let $H,K\subset\H$ standard subspaces  and $E$ a projection satisfying the assumptions of Lemma~\ref{lem:ort2}  w.r.t.\ $ H$ and $K$. Assume that $K\subset H'$. Then $EK\subset EH'$ and  $(1-E)K\subset (1-E)H'$.
\end{lemma}
\begin{proof}
Since $H'=EH'\oplus (1-E)H'$ and $K=EK\oplus (1-E)K$, for every $\xi\in K$ we have $E\xi\in K$, thus $E\xi\in EH'$. We conclude that $EK\subset EH'$ and analogously $(1-E)K\subset (1-E)H'$.
\end{proof}
\vspace{-0.7cm}
\subsection{One particle nets}\label{One-particle-nets-sub}

Let $U$ be a unitary representation  of the Poincar\'e group ${\pc\Poi}$ on  a Hilbert space $\H$. 
We shall call a \textbf{$U$-covariant (or Poincar\'e covariant) net of standard subspaces on wedges}  a map $$H:\W\ni W\longmapsto H(W)\subset\H,$$
associating to every wedge in $\RR^{1+3}$ a closed real linear subspace of $\H$, satisfying the following properties\footnote{{\bc The notation $W_1, W_2$ in this definition should not be confused with the standard wedges in the direction of particular axes, as used in Sect.~\ref{local-nets}.} }:
 \begin{enumerate}
\item \textbf{Isotony:}  If $W_1,W_2\in\W$ and $W_1\subset W_2$ then $H(W_1)\subset H(W_2)$;
\item \textbf{Poincar\'e  covariance:}   $U(g)H(W)=H(gW),$ $\forall g\in{\pc\Poi},\,\forall W\in\W$;
\item \textbf{Positivity of the energy:} the joint spectrum of translations in $U$ is contained in the forward lightcone $V_+=\{\, p\in\RR^{1+3}:  p^0\geq 0,  p^2=(p,p)\geq0\}$;
\item \textbf{Cyclicity:} if $W\in \W$, then  $H(W)$ is a cyclic subspace of $\H$;
\item \textbf{Locality:} if $W_1\subset W_2'$ then $  H(W_1)\subset H(W_2)'.$
\end{enumerate}

We shall indicate a $U$-covariant net $H$ of standard subspaces on wedges satisfying 1.-5.  with the couple $(U,H).$
This is the setting in which we are going to study the following property:
\begin{itemize}
\item[6.] \textbf{Bisognano-Wichmann property}: if $W\in\W$, then $U(\lambda_W(2\pi t))=\Delta^{-it}_{H(W)},$ $\forall t\in\RR;$
\end{itemize}

\nin The next property is a completeness property for a model  in the sense of the  causal structure and, by Lemma \ref{inc},  is a consequence of the locality  and the B-W properties (see e.g. \cite{Mor}).
\begin{itemize}
\item[7.] \textbf{ Duality property}: if $W\in\W$, then $H(W)'=H(W')$.
\end{itemize}
 Denote by $P^0, \mathbf{P}$ be the generators of translations in the representation $U$ and $M=\sqrt{(P^0)^2-\mathbf{P}^2}$ the resulting mass operator.
Then Theorem~\ref{Borch} has the following corollary, which is well known in the context of nets of von Neumann algebras.
\begin{corollary}\label{Borchers-corollary}  For any wedge $W$, the mass operator $M$ commutes strongly\footnote{{\bc Taking anti-linearity of $J_{H(W)}$ into account, commutation with \emph{real} bounded Borel functions of $M$ is understood here.}} with $\De_{H(W)}$ and $J_{H(W)}$. 
Its {\bc real} bounded {\bc Borel} functions commute weakly with $S_{H(W)}$ on domains specified as in (\ref{detailed-commutation}).
\end{corollary}
\begin{proof} Consider the wedge $W_1$, {\bc defined as in  Sect.~\ref{local-nets}}, and the associated standard subspace $H(W_1)$. Translations in {\bc direction of the axes $x_2$ and $x_3$}  fix $W_1$. In particular the generators of the associated translation group $P_2$ and $P_3$, respectively, commute  strongly with 
$\Delta_{H(W_1)}$ and $J_{H(W_1)}$ by Lemma~\ref{lem:sym}. Lightlike translations of the form $a_\pm(t)=(\pm t, t, 0,0)$ with $t\geq0$ have  generators $P_\pm:=( \pm P_0-P_1)$ {\bc s.t. $\pm P_{\pm}\geq 0$}  and $U(a_\pm (t))H(W_1)\subset H(W_1)$ for $t\geq0$. By the Borchers theorem  for standard subspaces (Theorem~\ref{Borch}) we have that $U(a_\pm)$ have the commutation relations as in equation \eqref{eq:bor}:
$$\Delta_{H(W_1)}^{is}U(a_\pm(t))\Delta_{H(W_1)}^{-is}=U(a_\pm(e^{\mp 2\pi s}t))\Rightarrow{\pc\Delta_{H(W_1)}^{is}f(P_\pm)\Delta_{H(W_1)}^{-is}=f(e^{\mp 2\pi s}P_\pm)}$$
$$J_{H(W_1)}U(a_\pm(t))J_{H(W_1)}=U({\bcc a}_\pm(-t))\Rightarrow J_{H(W_1)}{f(P_\pm)^{{\bcc *}}} J_{H(W_1)}=f(P_\pm),$$
where $f$ is any bounded Borel   function. {\bcc The implications above follow by approximating $f$ pointwise with Schwartz-class functions
(which gives strong convergence of the corresponding operators) and using the Fourier transform.  Now $P^2=M^2=-(P_+P_-+P^2_2+P_3^2)$ and for any real Borel function $g$ it is easy to check, using the above relations, that
$g(M^2)$ commutes with $\De_{H(W)}$ and $J_{H(W)}$. Indeed, by approximating $g$ pointwise by Schwartz-class functions, applying the Fourier
transform and using that $P_2,P_3$ commute strongly with $\De_{H(W)}$ and $J_{H(W)}$ it suffices to verify that   
\beqa
\Delta_{H(W_1)}^{is}e^{-i P_+P_- t} \Delta_{H(W_1)}^{-is}=e^{-i P_+P_- t}, \quad  J_{H(W_1)}e^{-i P_+P_- t} J_{H(W_1)}=e^{i P_+P_- t}. \non
\eeqa
This is achieved by approximating $e^{-i P_+P_- t}$ pointwise by linear combinations of expressions of the form $ f_+(P_+)f_-(P_-)$, where $f_+,f_-$
are bounded Borel functions,   and applying the relations above.
}

For a general wedge $W$, let $g\in\mathcal P_+^\uparrow$ s.t. $W=gW_1$. {\bc Then},   by Lemma~\ref{lem:sym},  $J_{H(W)}= U(g)J_{H(W_1)}U(g)^*$,  $\De_{H(W)}= U(g)\De_{H(W_1)}U(g)^*$ and thus  $S_W=U(g)S_{H(W_1)}U(g)^*$. Clearly $P^2=U(g)P^2U(g)^*$, thus $P^2$ commutes with $J_{H(W)}, \De_{H(W)}, S_{H(W)}$ for every $W\in\W$ in the same sense as discussed above. \end{proof}

\subsection{Induced representations:  the  Poincar\'e  group and  its sub-groups}\label{masslessrep}

 Our group theoretic considerations in the remaining part of Sect.~\ref{Preliminaries-section} and in Sect.~\ref{modularity-sect} are based on the  
\textbf{standing assumption} that all the representations of topological groups on Hilbert spaces are strongly continuous.

Let $G$ be a locally compact group, $N$ a nontrivial closed normal abelian subgroup and $H$ another closed subgroup such that 
$G=N\rtimes H${\footnote{We warn the reader, that the letter $H$, used earlier for nets of standard subspaces, is now used for groups. As we will not consider nets of standard subspaces in the remaining part of Sect.~\ref{Preliminaries-section}, there is no risk of confusion.}}.
 Assume that the action of $G$ on $\hat N$, the dual group of $N$, obtained by conjugation, is regular (cf. \cite{foll} Sect.~6.6 and Definition~\ref{regular-definition}). Let $p\in\hat N$, $\Omega_p$  be the orbit under the $G$-dual action\footnote{\pc $\chi_{g^{-1} p}(x)=\chi_p(gxg^{-1})$.}, with $x\in N$, $p\in\hat N$ and $g\in G$, $\Stab_p$ and $\overline{\Stab}_p$ be the stabilizers of the point $p$ under the action of $H$ and $G$. $\Stab_p$ is called
the \emph{little group}.
{\pc Let $\chi_p$ be the character associated to $p\in\hat N$.}

Every unitary irreducible representation of $G$  is obtained by induction in the following way  (see e.g.  \cite{foll} Sect. 6)
\begin{equation}\label{eq:ind}
\Ind_{\overline{\Stab}_p}^G ({\pc\chi_p}\cdot V),
\end{equation} 
where $V$ and  ${\pc \chi_p}\,\cdot \,V$ are unitary representation{\bcc s} of the little group $\Stab_{p}$ and of $\overline{\Stab}_p$, respectively, and the  following proposition holds:

\begin{proposition}\cite{foll} \label{Foll}
Let $G=N\rtimes H$ as above. Every unitary {\pc irreducible} representation of $G$ is equivalent to one of the form \eqref{eq:ind}. Furthermore $\Ind_{\overline{\Stab}_p}^G ({\pc \chi_p}\cdot V)$ and $\Ind_{\overline{\Stab}_q}^G ({\pc \chi_q}\cdot V')$ are equivalent if and only if $p$ and $q$ belongs to the same orbit, say ${\pc p=g\, q}$, and $V$ and {\pc$V'\circ ad_{g^{-1}}$} are equivalent representations of $\Stab_p$ {for some $g\in G$}. 
\end{proposition}
{\pc If $W=\Ind_{\overline{\Stab}_p}^G (\chi_p\cdot V)$ is an irreducible representation of $G$ then the spectral measure of $W|_N$ is concentrated on the orbit $o=Gp$ (cf. Proposition 6.36 \cite{foll}).}

\nin References for general induced representations are for instance  \cite{foll, kir, Bar}. \\

{\bf The Poincar\'e group.} {\pc The Minkowski space $\RR^{1+3}$ is the 4-dimensional real vector space endowed with the  metric tensor $\eta=\mathrm{diag}(1,-1,-1,-1)$.
The Lorentz group $\L$ is the group of linear transformations $L$ s.t. $L^T\eta L=\eta$.} Let $\L_+^\uparrow$  {\bc be} the connected component of the identity of the Lorentz group and $\widetilde \L_+^\uparrow=\SLC$ its universal covering group. Let $\tPoi=\RR^{1+3} \rtimes \SLC$ 
be the universal covering of the Poincar\'e group  $\Poi=\RR^{1+3} \rtimes \L_+^\uparrow$ (the inhomogeneous symmetry group of $\RR^{1+3}$)
and $\Lambda$ be the covering map.  First of all, we recall that to any 4-vector is 1-1 associated a $2\times2$ matrix {\pc
$$ x_\sim= x_0\,\eins+\sum_{i=1,2,3}x_i\,\sigma_i=\left(\begin{array}{cc}
 x_0+x_3&x_1-ix_2\\x_1+ix_2&x_0-x_3
 \end{array} \right),$$ where $\sigma_i $ are the Pauli matrices. Real vectors define Hermitian matrices.
 If $A\in\SLC$, and $\Lambda: \tPoi\rightarrow \Poi$  is the covering homomorphism, then the Poincar\'e action is ruled by the following relation
\beqa 
({\Lambda(A)x})_\sim=A x_\sim A^*. \label{action-p}
\eeqa

Let  $U$ be a unitary strongly continuous representation of the Poincar\'e group, then the representation of an $x$-translation has the form $U(x)=e^{iPx}$ where $P$ is a vector of four self-adjoint operators and $Px$ is obtained through the Minkowski product.  Every $g\in\widetilde \L_+^\uparrow$ acts on an $x$-translation   by  the adjoint action, namely $gxg^{-1}=\Lambda(g)x$.
Let $\mathrm{Sp}\,(P)$ be the {\bcc joint spectrum of generators of translations} and $p$ be  a point in the spectrum, then we have the character $\chi_p(x)=e^{ipx}$. 
As in the general case, the dual action on the momentum space is defined s.t. $\chi_p(\Lambda(g)\cdot x)=\chi_{p'}(x)$ and it is easy to see that $p'=\Lambda(g)^{-1}p$, where the latter is the matrix-vector multiplication. Clearly the adjoint action of the translations act trivially on themselves, hence on their dual {\pc(see e.g. \cite{Bar})}. 
}

{\bf  Positive energy massless representations  of the  Poincar\'e group}. Let  {\pc $\chi_q$}, {\bcc $q\neq 0$}, be a character of {\bc the} translation group.  
  We shall call $\Stab_q$ and $\overline\Stab_q$ the stabilizers of the point $q$ through the $\widetilde \cL_{+}^\uparrow$ and $\tPoi$ actions, respectively. The latter is  $ \overline{\Stab}_q=\RR^{1+3}\rtimes\Stab_q$, where $\Stab_q$ shall be called {as above} the little group. 
Any massless, unitary, positive energy representation   of   $\tPoi$  is obtained starting with the character  associated to
$
q:= (1,1,0,0)\in \partial V_+
$  
($\partial V_+\setminus\{0\}$ is {\bc an} $\cL_+^\uparrow$-orbit) and inducing by a unitary representation of the $\overline\Stab_q$ group. 
Note that a $\overline{\Stab_q}$ representation is of the form $$\RR^{1+3}{\pc \rtimes}\Stab_q\ni (x,\sigma)\mapsto {\pc\chi_q}(x)V(\sigma),$$ where $V$ is the unitary representation of  $\Stab_q$. 
{\bc The little group $\Stab_q $ is isomorphic to $\widetilde \EE(2)=\RR^2 \rtimes \mathbb T$, where $\mathbb T$ is the unit circle. {\pc Note that $r_1(\theta)$ generate $\mathbb T.$} $\widetilde \EE(2)$ is the double cover of  $\EE(2)=\RR^2 \rtimes \mrm{{\gc S}O}(2)$, which is the 
group of Euclidean motions in two dimensions\footnote{We note that $\widetilde \EE(2)$ is not the universal covering group of $\EE(2)$, as it is not simply connected. {\pc In particular the covering map {\bcc is given by} $\mathbb T\ni r_1(\theta)=e^{i\frac\theta2\sigma_1}\mapsto e^{i\theta }\in\SO(2)$, where we have identified $\RR^2$ with $\CC$ and rotations as multiplication by a phase $e^{i\phi}$.}}.
Irreducible representations $V$ of $\widetilde \EE(2)$  fit in one of the following two classes: (See e.g. \cite{var} and \cite[page 520]{Bar})
\begin{itemize}
\item[$(a)$] The restriction of $V$ to $\RR^2$ is trivial;
\item[$(b)$] The restriction of $V$ to $\RR^2$ is non-trivial.
\end{itemize}
{\bc Irreducible representations of $\widetilde \EE(2)$ in class $(a)$ are 
labelled by {\bc half-}integers $h$,  called the {\it helicity} parameters. } 
Irreducible
representations in class $(b)$ are labelled by $\kappa >0$, the radius
of a circle in $\RR^2$, namely the joint spectrum of the $\widetilde \EE(2)$-translations, and a Bose/Fermi alternative parameter $\epsilon\in\{0,\frac12\}$. 

Let $$U=\Ind_{\overline\Stab_q  }^{{ \tPoi}}{({\pc\chi_q}\cdot V)}$$   be a unitary representation of $\tPoi$ induced from the representation
${\pc\chi_q}\cdot V $ of $\overline\Stab_q$.  We say that $U$ has \emph{finite helicity} or \emph{infinite spin} if $V$ has the form $(a)$ or $(b)$, respectively.

An irreducible \textbf{finite helicity} representation is of the form
\beqa
U_{h}=\Ind_{\overline\Stab_q}  ^{{ \tPoi }}  ({\pc \chi_q}\, \cdot\, V_{2h}) ,\qquad h\in\frac{\ZZ}{2},    \label{eq:U}
\eeqa
{\bc where $V_{2h}(y,g)=(2h)(g)$, $(y,g)\in  \RR^2 \rtimes  \mathbb T$  and ${2h}$ is the one dimensional representation of  $\mathbb T$ of character  
$h\in \Z/2$. {\pc In particular, $r_1\in\Stab_q$ and $(2h)(r_1(\varphi))=e^{i\fr{\varphi}{2} 2h}$ for a $\varphi\in [0,4\pi]$.} Integer and half-integer values of $h$ discern, respectively, bosonic and fermionic representation of $\widetilde \EE(2)$, hence, by induction, of $\tPoi$.
\vspace{0.5cm}
 
\subsection{$G_W$  and related  subgroups  of the Poincar\'e group and their representations}

{\bc In this subsection we introduce certain subgroups of $\tPoi$ which will be needed in Sect.~\ref{modularity-sect} below to formulate
a criterion for the B-W property}.
 We refer to Sect.~\ref{local-nets} for the definitions of the wedges $W_i$, $i=1,2,3$, and the set of wedges $\W$. 
We also recall from this subsection that   $R_i, \La_i\in \mcL_+^{\uparrow}$ and  $r_i, \la_i\in \mathrm{SL}(2,\CC)$ denote the one-parameter families of rotations and boosts preserving 
the wedges $W_i$, $i=1,2,3$.
\begin{definition}\label{G-definition}
We denote with 
\begin{itemize}
\item $G_3^0$ the subgroup of $A\in\SLC$ s.t. $\Lambda(A)W_3=W_3$.
\item $G_3=\langle G_3^0,\RR^{1+3}\rangle$, where $\RR^{1+3}$ is the translation  group and  $\langle G_3^0,\RR^{1+3}\rangle$  
 denotes  the  group generated by $G_3^0$ and $\RR^{1+3}$.
\item  $\tilde G_3^0=\langle G_3^0, r_1(\pi)\rangle$,  $\tilde G_3=\langle G_3, r_1(\pi)\rangle =\RR^{1+3} \rtimes \tilde G_3^0$.
\end{itemize}
For a general wedge  $W\in\W$, $G_W^0$, $G_W$, $\ti G_{W}$ and $\ti G_{W}^0$ are defined  by the transitive action of $\Poi$ on wedges.
We will denote the massless orbits of the $\RR^{1+3}$ translation characters under the $G_3$ action with $\sigma_r=\{p=(p_0,\textbf{p}= (p_1,p_2,p_3))\in\RR^{1+3}: 
p_1^2+p_2^2=r^2, p_0^2=\textbf{p}^2, {\gc p_0>0}\}$, $\sigma_0^\pm=\{p\in\RR^{1+3}: p_1=0=p_{2}, p_0=\pm p_3, {\gc p_0>0}\}$ {\pc and $\sigma_0=\sigma_0^+\cup\sigma_0^-$}. {\bc In the present paper, we will be interested in  orbits  $\si_r$ with $r>0$ since $\sigma_0^\pm$  have null measure w.r.t. the Lorentz invariant measure on $\partial V_+$.}
 {\bc Finally, we warn the
reader that $\tilde G_3, \ti G_{W}$ do not denote the covering groups of $G_3, G_{W}$.}
\end{definition}
Note that $G_3^0=\langle \lambda_3,\,r_3, r(2\pi)\rangle$.   
 Indeed, any $\mathrm{SL}(2,\CC)$ {\bc element} implementing a Poincar\'e transformation can be decomposed by the polar decomposition  $A=U_A\cdot T_A$ (see e.g. \cite{Morlect}),  where $U_A$ is a rotation and $T_A$ a boost. Let $\Lambda:\mathrm{SL}(2,\CC)\rightarrow 
  {\bc \mathcal{L}_+^{\uparrow}}$  be the covering map, then $\Lambda (A)W_3=W_3$ iff $ U_A ^{-1}W_3=T_AW_3$.  Assume that there exists a transformation $A$ such that $U_A ^{-1}W_3
\neq W_3\neq T_AW_3$ but $U_A ^{-1}W_3=T_AW_3$. Consider the edge of the wedge $E{\bc :=}\{x\in\RR^{1+3}: x_0=0=x_3\}$. Then,
$U_{A}^{-1}E=\{x\in\RR^{1+3}: x_0=0=(U_{ A}  x)_3\}$ cannot be equal to $T_{A} E=\{x\in\RR^{1+3}: ( T^{-1}_{A} x)_0=0=( T_{A}^{-1}x)_3\}$. In particular ${\bc U_A ^{-1}}W_3= W_3=T_AW_3$.

Next, we note that $G_3^0$ and $\tilde G_3^0$ share the same orbits in ${\bc \pa} V_+$ as the following remark explains. 
\begin{remark}\cite{Mor} \label{rmk:orb} Fix $p=(p_0,p_1,p_2,p_3)\in (\partial V_+  \setminus  \{0\}  ):=\{p\in\RR^{1+3}:p^2=0, p_0>0\}$. The $R_1(\pi)$-rotation  $$R_1(\pi)p=(p_0,p_1,-p_2,-p_3)$$ can be obtained as a composition of a $\Lambda_3$-boost of parameter $t_p$ and a $R_3$-rotation of parameter $\theta_p$ as
\begin{equation}\label{'''}\Lambda_3(t_p)R_3(\theta_p)(p_0,p_1,p_2,p_3)=\Lambda_3(t_p)(p_0,p_1,-p_2,p_3)=(p_0,p_1,-p_2,-p_3)\end{equation} for all the orbits except 
 for {\bc  $\si_0^{\pm}$, 
  where $\si_0^{\pm}$ appeared in Definition~\ref{G-definition}}.     
Clearly $t_p$ and $\theta_p$ depend on $p$ and the orbits excluded by this geometrical fact have null measure  w.r.t. the Lorentz invariant measure on
$\partial V_+$. {\pc The discussion does not change if $p$ is considered as an element of $\RR^{1+3}$ or of its dual.}
\end{remark}

By \eqref{'''}, we deduce that  almost all $G_3^0$ orbits on $\partial V_+$ are preserved by the $R_1(\pi)$-action. Furthermore, we note that $R_1(\pi)$ sends $W_3$ onto $W_3'$. Thus any transformation $R\in\Poi$ such that $RW_3=W_3'$ also preserves the $G_3^0$ orbits on $\partial V_+$ as well as $R_1(\pi)$, since $R_1(\pi)R\in G_3^0$. 
We have just seen that it is possible to pointwise reconstruct a transformation sending $W$ to $W'$ just starting with elements in $G_W^0$. 
With the help of the modularity condition (cf. Definition~\ref{def:mod} and Theorem~\ref{thm:starcor} below),  this gives the proof of  the B-W property in the scalar case in \cite{Mor}. 

{\bc As the regularity of the action of $G_3^0$ and $\ti G^0_3$ on $\real^{1+3}$ is verified in Appendix~\ref{App-C}, we can apply the theory of induced representations to $G_3=\RR^{1+3}\rtimes G^0_3$ and $\ti G_3 = \RR^{1+3}\rtimes \ti G^0_3$.}
 Choose a point $q_r$ on each massless, {\bc positive energy} orbit $\sigma_r$  of $G_3^0$ on the dual of $\RR^{1+3}$. Up to a null measure set in $\partial V_+$,  the stabilizer  of  $q_r$  in $G_3$ is $\RR^{1+3}\times  \lan r(2\pi) \ran$ (cf.~{\bc the} definition of the orbits  $\si_r$ and of $G_3^0$). 
Thus there exist  only two irreducible representations of $G_3$ induced by $\chi_{q_r}$, namely 
\begin{equation}\label{eq:tG}W_{r,n}=\displaystyle{\Ind_{\RR^{3+1}\times  \lan r(2\pi)  \ran}^{G_3}(\chi_{q_r}\cdot V_n )},\end{equation} where $V_n(r(2\pi))=(-1)^n$ and $n=0,1$, cf. Proposition~\ref{Foll}.  They correspond to bosonic and fermionic representations of $G_3$.

Again by  Remark \ref{rmk:orb} the subgroups $\tilde G_3$ and $\tilde G_3^0$  share the same orbits  $\sigma_r$ with $r>0$ on $\partial V_+$  (up to a null measure set in $\partial V_+$).
On the other hand the stabilizer of the point $q_r=(r,r,0,0)\in\sigma_r$ in $\tilde G_3$ is $\RR^{1+3}\rtimes  \lan r_1(\pi) \ran$. 

 Thus the little group of $q_r$, the subgroup of $ \tilde G_3^0$ fixing $q_r$,  is $\ZZ_4$. We have four irreducible representations of $\ZZ_4$ indexed by the representation of the generator, namely $V_n(r_1(\pi))=i^n$, with $n=0,1,2,3.$
Correspondingly, we have four induced representations of $\tilde G_3$ associated to  each orbit, namely
\beqa W_{r,n}=\Ind_{\RR^{1+3}\rtimes \ZZ_4}^{\tilde G_3} (\chi_{q_r}\cdot V_n ) \label{W-rep}\eeqa
acting on the Hilbert spaces $\hil_{r,n}$.
 Note that we called $V_{2h}$ the representation of character $ {2h}$, trivial on translation{\bc s} of $\widetilde \EE(2)$, and $V_n$ the representation of character $n$ of $\ZZ_4$. That this  is not an abuse of notation is justified by the fact that $r_1({\pi})\in \widetilde \EE(2)$ and when we restrict $V_{2h}$ to the group  $\langle r_1(\pi)\rangle$, we get
$$V_{2h}(r_1(\pi))=e^{i\frac\pi2 2h}=i^{2h}=V_n(r_1(\pi)),\quad V_{2h}(r(2\pi))=(-1)^n=V_n(r(2\pi))$$ and it is enough to consider $2h$ in $\ZZ_4$ in the first case or $2h$ in $\ZZ_2$ in the second. 
Furthermore $W_{r,n}$ as a representation of $\tilde G_3$ restricts to $W_{r,m}$ with  $m\,{\pc \equiv}\, n\, (\mathrm{mod} \,\, 2)$ as a representation of $G_3$, cf. Lemma \ref{lem:rest} in the next section.
 We will refer to \eqref{eq:tG} and \eqref{W-rep} {\bc as} {\it massless} representation  of $G_3$ and $\tilde G_3$, respectively, since the translation spectrum lies on the boundary of the forward lightcone.
\section{Modularity condition and $U_h$ restriction} \label{modularity-sect}
\setcounter{equation}{0}

In this section we show that any local net of standard subspaces, covariant under  a finite or infinite multiple of $U_{h}\oplus U_{-h}$, $h\in \mathbb{Z}$, satisfies the B-W property. 
The analysis is based on the following modularity condition:
\begin{definition}\cite{Mor}\label{def:mod} 
A (unitary, positive energy) $\tPoi$ representation is said to be \textbf{modular }if, for any $U$-covariant net of standard subspaces $H$, we have that $(U,H)$ satisfies the B-W and duality properties.

Let $W\in\W$.  A unitary, positive energy $\tPoi$-representation  $U$ satisfies the {\bf modularity condition} if for an element $r_W\in\tPoi$ such that $\Lambda(r_W)W=W'$  we have that
\begin{equation}\tag{MC}\label{eq:cond}{U(r_W)\in U(G_W)''}. \end{equation}
\end{definition}
Note that \eqref{eq:cond} depends neither on the choice of $r_W$ nor on W. Indeed if $\tilde r_W\in\tPoi$ is another transformation such that $\Lambda(\tilde r_W) W=W'$ then $ r_W\cdot\tilde r_W\in G_W$ and if \eqref{eq:cond} holds for $U(r_W)$, then it holds for $U(\tilde r_W)$. By  transitivity of  the $\tPoi$ action on wedges it is not restrictive to fix a wedge region $W$. Condition \eqref{eq:cond} can be straightforwardly stated when just a representation of $\tilde G_3$ is taken into account.
\begin{theoreme}\cite{Mor}\label{thm:starcor}
Let $U$ be a positive energy unitary representation of the Poincar\'e group $\tPoi$.  If condition \eqref{eq:cond} holds for $U$, then  any local $U$-covariant net of standard subspaces $H$, namely any pair $(U,H)$, satisfies B-W  and duality properties. In particular $U$ is modular.
\end{theoreme}
\nin {\bcc We remark that  the class of nets in Theorem~\ref{thm:starcor}, transforming under $\tPoi$, is more general than the class of bosonic nets we defined in 
Sect.~\ref{One-particle-nets-sub}.}  
The theorem applies to the families of Poincar\'e representations  covered by the following two results.
\bep\cite{Mor}\label{prop:finMC}
Let $U$ be an irreducible scalar massive or  massless representation of the Poincar\'e group, then $U$ satisfies \eqref{eq:cond}.
\eep
The proof of the above proposition adapts in irreducible finite helicity case. This section contains an alternative independent proof of modularity condition for finite helicity representations, see Corollary \ref{modularity-finite-helicity}.
\begin{proposition}\label{prop:P} Let $U$, $U_1$ be unitary positive energy representations of $\tilde G_3$ 
 on $\H$, satisfying \eqref{eq:cond}.
\begin{enumerate}
\item[(i)] Let $E$ be {\bcc the} projection on the subspace, where $U(r(2\pi))=1$. {\bcc Then} $U$ satisfies \eqref{eq:cond} iff both $EU(\cdot)E$ and $(1-E) U(\cdot)(1-E)$ satisfy \eqref{eq:cond}.
\item[(ii)] Let $\K$ be a Hilbert space, then \eqref{eq:cond} holds for $U\otimes 1_\K\in\B(\H\otimes \K)$.

\item[(iii)] Let $U_2$ be a  unitary representation of $\tilde G_3$ s.t. $U_2$ is unitarily equivalent to $U_1$. 
Then $U_2$ satisfies \eqref{eq:cond}.
\end{enumerate}
\end{proposition}

\begin{proof} (i) Since $r(2\pi)$ commutes with every element in $\tilde G_{3}$, the spectral projection $E$ of  $U(r(2\pi))$  decomposes 
$U$ into disjoint representations $U(\cdot)=EU(\cdot)\oplus (1-E) U(\cdot)$. The thesis follows since
$U(G_3)''=(EU(G_3))''\oplus ((1-E)U(G_3))''$. 

(ii) Is proved in \cite{Mor}. 

(iii) Let $W$ be a unitary s.t. $WU_2(g)W^*=U_1(g)$, $g\in\tilde G_3$.
It is easy to see that $WU_2(G_3)'W^*=U_1(G_3)'$ and  $WU_2(G_3)''W^*= U_1(G_3)''$.  If $WU_2(r_1(\pi))W^*=U_1(r_1(\pi))$, then 
$$U_1(r_1(\pi))\in U_1(G_3)''\Rightarrow WU_2(r_1(\pi))W^*\in WU_2(G_3)''W^*\Rightarrow U_2(r_1(\pi))\in U_2(G_3)'',$$
which concludes the proof. \end{proof}
The  representations $W_{r,n}$ of $\ti G_3$, restricted to $G_3$, give $W_{r,m}$, {\bcc where} $m\,{\pc\equiv}\,n (\mrm{mod}\,\, 2$), see Lemma~\ref{lem:rest} below. Thus they are disjoint for different $r$,  by the disjointness of  the respective orbits~$\si_r$,  cf. Proposition \ref{Foll}.
 Furthermore, since $W_{r,n}$ is irreducible,  $W_{r,n}(G_3)'=\CC\cdot \textbf{1}$  and  $W_{r,n}(r_1(\pi))\in W_{r,n}(G_3)''=
B(\hil_{r,n})$.  We deduce that $W_{r,n}$ satisfies (\ref{eq:cond}).
\begin{corollary} \label{Mod-Corollary}   Let $W_{r,n}$ be the irreducible ${\ti G_3}$-representation of radius $r>0$. Then $W_{r,n}$ satisfies \eqref{eq:cond}.
\end{corollary}
 The following proposition ensures that also a direct integral of massless $\tilde G_3$ representations satisfies \eqref{eq:cond}. 
\begin{proposition}\label{prop:PP}
Let $\mu$ be a positive Borel measure on $\RR^+$.
Assume that  $U_{r}$ are multiples of  {\bc the massless $W_{r,n}$, $n=0,1,2,3$,  representations  of $\ti G_{W}$}.  Then $U=\int^{{\bc \oplus}}_{\RR^+} U_{r} d\mu(r)$ satisfies \eqref{eq:cond}.
\end{proposition}
One can reduce the argument to the cases {\bc  $U_r|_{G_W}=W_{r,0}$ or $U_r|_{G_W}=W_{r,1}$ } by {\bc Proposition}~\ref{prop:P}~(i). We give details on  the proof in  Appendix~\ref{proposition}.

Now we want to verify the modularity condition \eqref{eq:cond} for a large family of massless bosonic representations {of the physically relevant form $U_h\oplus U_{-h}$ for integer helicities $h$}. In order to prove the result, we want to disintegrate the restriction of $U_h$ to the $\tilde G_3$ subgroup, to check the condition \eqref{eq:cond} on the disintegration and to apply Theorem~\ref{thm:starcor}. 
The Poincar\'e representations are obtained by induction and the Mackey subgroup theorem teaches how to make the disintegration for such kind of representation. 

Let $H_1$ and $H_2$ be subgroups of a locally compact group $G$.  Then $H_1\backslash G\slash H_2$ is the   double coset, i.e., the set of the equivalence classes $[g]=H_1gH_2$, with $g\in  G$. 

\begin{definition}\label{regularly-related-definition}\cite{ma52} Let $G$ be a separable locally compact group.

Closed subgroups $H_1$ and $H_2$ of
$G$ are said to be {\it regularly related} if there exists a sequence $E_0, E_1, E_2, \ldots$ of measurable subsets of $G$ each of which is a union of double cosets in $H_1\backslash G\slash H_2$  such that $E_0$ has Haar measure zero and each double coset  not in $E_0$ is the intersection of the $E_j$ which contain it.

Because of the correspondence between orbits of $G/H_{\bc 2}$ under
$H_{\bc 1}$ and double  cosets $H_1\backslash G\slash H_2$, $H_1$ and $H_2$ are regularly related if and only if the orbits, (i.e., the double cosets) outside of a certain set of measure zero form the equivalence
classes of a measurable equivalence relation. Given a topological standard measure space $X$, an equivalence relation $\sim$ and the quotient map $s: X\to Y=X/\sim$, the equivalence relation  is said to be measurable if there exists a countable family $\{F_n\}_{n\in\NN}$ of subsets of the quotient space $Y$, s.t. $s^{-1}(F_n)$ is measurable and   each point in $Y$ is the intersection of all the $F_{n'}$, $n'\in\nat$, containing this point.

Consider the map $s: G\rightarrow H_1\backslash G\slash H_2$  carrying each element of $G$ into its double coset. Then equip $H_1\backslash G\slash H_2$ with the quotient topology given by $s$ and consider a finite measure $\mu$ on $G$ which is in the same measure class\footnote{{Has the same set of null measure.}}  as the Haar measure. It is possible to define  a measure $\bar\mu$ on the  Borel sets of $H_1\backslash G\slash H_2$ by $\bar\mu(F)=\mu(s^{-1}(F))$. We shall call $\bar \mu$ an \textit{admissible measure} in $H_1\backslash G\slash H_2$. The definition is well posed since any two of such measures have the same null measure sets.

\end{definition}

General theory of induced representations can be found for instance in \cite{ma52, foll, Bar}. We recall the Mackey's subgroup theorem. 
\begin{theoreme}[Mackey's subgroup  theorem]\label{thm:sgr}\cite{ma52}.
Let $H_1,\, H_2$ be closed subgroups regularly related in $G$. Let $\pi$ be a strongly continuous representation of $H_1$.   
For each $g\in G$ consider $H_g=H_2\cap(g^{-1}H_1g)$ and set
$$\mathcal{V}_g=\mathrm{Ind}_{H_g}^{H_2}(\pi\circ\ad\, g).$$
Then $\mathcal{V}_g$ is determined to within equivalence by the double coset
$[g]$ to which $g$ belongs. If $\nu$ is an
admissible measure on $H_1\backslash G\slash H_2,$ then
\beqa (\mathrm{Ind}_{H_1}^G\pi )|_{H_2}\simeq\int^{\oplus}_{H_1\backslash G\slash H_2}\mathcal{V}_{[g]}\,d\nu([g]).\label{Mackey-formula}
\eeqa
\end{theoreme}
 An immediate application of  Theorem \ref{thm:sgr} to the restriction of $W_{r,n}$ to $G_3$ is the following lemma, which entered into the proof of Corollary~\ref{Mod-Corollary} above.
\begin{lemma}\label{lem:rest} The restriction of the $\tilde G_3$ representation $W_{r,n}$ to $G_3$ is $W_{r,m}$, where $\displaystyle{m\,{\pc \equiv}\, n\, (\mathrm{mod}\,\,  2)}.$
\end{lemma}

 \nin Of central importance for our analysis is the following proposition.
\begin{proposition}\label{Corollary-of-Mackey}
$\displaystyle{U_h|_{\tilde G_3}\simeq \int_{\RR^+}^{\oplus}W_{r,2h}dr}$,  where $2h$ in the right-hand side has to be considered modulo $4$. 
\end{proposition}
\begin{proof}
The $h$-helicity representation  $U_h$ is induced by the stabilizer $\overline\Stab_{q_1}$ of the point ${q_1}=(1,1,0,0)$.  Again $\overline\Stab_{q_1}$ is isomorphic to 
$\RR^{1+3} \rtimes  \widetilde \EE(2)$.
 In Theorem \ref{thm:sgr} we can consider $G=\tPoi$, $H_1= \RR^{1+3}\rtimes \widetilde \EE(2)$ and we want to study the restriction of $U_h$ to $H_2=\ti G_3$. We postpone the proof of the fact that $H_1$ and $H_2$ are regularly related. 

 {\bcc Let us now compute $H_g$ and $\mathcal{V}_g $ for several choices of $g$.   
First, for $g=\eins$,  $H_{g=\eins}=\RR^4\rtimes \langle r_1(\pi)\rangle\subset\RR^{4}\rtimes \widetilde \EE(2)$. Here we made use of the fact that 
$\ti G_3^0=\langle \lambda_3,\,r_3, r_1(\pi)\rangle$  and $ r_1(\pi)\in \Stab_{q_1}=\widetilde \EE(2)$.
Hence }
$$\mathcal{V}_{\eins}=\ind_{\RR^4\rtimes \langle r_1(\pi)\rangle}^{\tilde G_3}(\chi_{q_1}{\bc \cdot}V_{2h})= W_{1,2h},$$ by \eqref{W-rep} and the comment below  this formula. 

{\bcc Next, we consider $g={\pc \lambda_1(\ln r)}\in\tPoi$, where $\lambda_1$ is the lift of $\La_1$,  the boost in the $x_1$-direction, to $\mrm{SL}(2,\complex)$, see Sect.~\ref{local-nets}.  Setting  $t\mapsto {\Lambda}_1(t)=\Lambda(\lambda_1(t))$ and noting that 
$q_r= \Lambda_1(\ln\,r){q_1}=r\cdot{q_1}$,    the intersection  $H_g:=H_2\cap(g^{-1}H_{1}g)$ satisfies}
{\pc $$H_g=H_{\eins}\;\;  \text{ since  }\;\; \ti G_3\cap \left({\lambda}_1\left({\ln r}\right){\bcc^{-1}}\,\overline{\Stab}_{q_1} \,{\lambda_1}\left({\ln r}\right)\right)=\ti G_3\cap\overline{\Stab}_{q_{r^{-1}}}=\RR^4\rtimes \langle r_1(\pi)\rangle.$$}
Hence,
\begin{align}\label{V-computation}
\begin{split}
\mathcal{V}_{\Lambda_1(\ln r)}&=\Ind_{\RR^4\rtimes \langle r_1(\pi)\rangle}^{\ti G_3}\left((\chi_{q_1}{\bc \cdot}V_{2h})\circ \ad\, {\bc \lambda_1}\left({\ln r}\right)\right)\\&=\Ind_{\RR^4\rtimes \langle r_1(\pi)\rangle}^{\ti G_3}(\chi_{q_{r^{-1}}}{\bc \cdot} V_{2h})=W_{r^{-1},2h}
\end{split}
\end{align}
{\pc because  ${\lambda_1}\left({\ln r}\right)$ commutes with $r_1(\pi)$ (see comment on dual action in Sect. \ref{masslessrep}).}

{\bcc Finally, we consider $g=r_2(-\pi/2)$. We note for future reference that  $\Lambda(g)^{-1} q_1\in\sigma_0$ since  
$\Lambda(r_2(\pi/2))q_1=R_2(\pi/2)q_1 =(1,0,0,1)=:q_0\in\sigma_0$. We have }
$$H_{r_2(-\pi/2)}=\tilde G_3\cap r_2(\pi/2)\,\overline \Stab_{q_1}\,r_2(\pi/2)^{-1}=\tilde G_3\cap \overline{\Stab}_{q_0}=\RR^{1+3}\rtimes \langle r_3\rangle, $$
$$\mathcal{V}_{r_2(-\pi/2)}=\Ind_{H_{r_2({\bcc -}\pi/2)}}^{\tilde G_3}\left((\chi_{q_1}{\bc \cdot}V_{2h})\circ \ad\,{r_2(\pi/2)^{-1}}\right)=\Ind_{H_{r_2({\bcc -}\pi/2)}}^{\tilde G_3}(\chi_{q_0} \cdot \chi_{2h})$$
and the joint spectrum of translations is supported in $\sigma_0$.  Here $\chi_{2h}$ is the $2h$--character representation of $r_{\bcc 3}(\,\cdot\,)$, $\chi_{2h}(r_{\bcc 3}(\theta))=e^{i\frac\theta2 2h}$, $h\in\frac\ZZ2$.

{\pc Let us now we show that $[{\lambda}_1(t)]$ and$[r_2(-\pi/{\bcc 2})]$ cover all the equivalence classes in $H_1\backslash G\slash H_2$. } 
 Let $g_1,g_2\in\tPoi$  and assume that ${\pc p}_a$ and ${\pc p}_b$ are points on the same {\bcc massless} ${\bc \ti{G}_3}$-orbit\footnote{We have already one representative in each orbit, see above.} s.t.  ${\pc \Lambda}(g_1^{-1}){q_1}= {\pc p}_a$ and $ {\pc \Lambda}(g_2^{-1}){q_1}= {\pc p}_b$.  That is, there exists $x\in {\bc \ti{G}_3}$  such that ${\pc p}_b={\pc\Lambda}(x) {\pc p}_a$. Then 
$g_1x^{-1}g_2^{-1}=\mrm{s}\in {\bc \ov{\Stab}_{q_1} }$, thus $g_2=\mrm{s}^{-1}g_1x^{-1}$ and $g_2$ belongs to  $[g_1]$, the double coset of $g_1$ in $H_1\backslash \tPoi\slash H_2$.  {\pc In particular, for every $g\in G$ {\bcc for which} there exists $t_g\in\RR$ {\bcc such that } ${\bcc \La}( g^{-1}){q_1}$ and $\Lambda_{1}(-\ln t_g){q_1}$ belong to the same orbit $\sigma_{t_g^{-1}} $, {\bcc we have} $g\in[ \lambda_1\left({\ln t_g}\right)]$. All the other $g\in G$ such that $\Lambda (g)^{-1}q_1\in \sigma_0$ belong to the  double coset  $[r_2(-\pi/2)]$}. {\bcc This follows from the fact that $\Lambda(r_2(-\pi/2))^{-1}q_1=
q_0\in\sigma_0$ which was mentioned above. } 
 
{\bcc Now we will verify that  the sets $H_1$ and $H_2$ are regularly related. 
For this purpose},  we identify $H_1\backslash G\slash H_2$ with $\RR^+{\pc=[0,+\infty)}$  when $[{\pc \lambda_1}\left({\ln r}\right)]\in   H_1\backslash G\slash H_2$ is identified with $r^{-1}\in\RR^+$ and $[r_2(-\pi/2)]$ with $0$.  {\pc We consider $\RR^+$ with the relative topology inherited from $\RR$.}
The quotient topology w.r.t. the map $s:G\to H_1\backslash G\slash H_2$ defined by the previous identification  guarantees that intervals of $\real^+$  are open sets. Indeed, let $\epsilon>0$ and take an element $g$ in the preimage {\bcc $s^{-1}\left(I\right)$, $I:=\RR^+\cap (a,b)$. Then} the set  $\mathcal {N}_g^{\epsilon}=\{g'\in G: 
{\bcc \La({g'}^{-1})}{q_1}\in B_\epsilon( {\bcc \La({g}^{-1})  }{q_1}\}$ is an open neighbourhood of $g$: 
 {\bcc as the inverse image of the open set $B_\epsilon( {\bcc \La({g}^{-1})  }{q_1})$ under the continuous mapping $g'\mapsto \La(g'{}^{-1})q_1$}, the set $\mathcal{N}^{\epsilon}_g$ is  open and contains $g$.  {\bcc For sufficiently small $\epsilon$} it is  contained in $s^{-1}{\pc (I)}$, since if 
${\bcc  \La({g}^{-1}) }{q_1}\in\sigma_r$,     then, for every $g'\in\N_g^{\epsilon}$,  ${\bcc \La({g'}^{-1})}{q_1}\in \sigma_{r'}$ with $r,r'\in {\pc I}$ by continuity of the Poincar\'e action on Minkowski space.
With the identification $H_1\backslash G\slash H_2\simeq\RR^{+}$, it is easy to see that $H_1$ and $H_2$ are regularly related by using  intervals $(q-\frac1n,q+\frac1n)\cap \RR^+$, $n\in \nat$, with rational center contained in $\RR^{+}$. 
 (The second part of Definition~\ref{regularly-related-definition} is used here).

{\bcc Let us now describe the equivalence class of an admissible measure (cf.~Definition~\ref{regularly-related-definition})}.
{\pc Starting with a finite measure $ \mu$ on $G$ in the equivalence class of the Haar measure, we induce a measure on $\RR^+ {\bc \simeq} H_1\backslash G\slash H_2$, which we prove to be in the  measure class of the Lebesgue measure. Indeed, let $W_r$ be the representation $W_{r^{-1},2h}$ and $W_0$ be $\V_{[r_2(-\pi/2)]}$, we get the formula:
\beqa 
U_{h}|_{\tilde G_3}\simeq \int_{\RR^+}^\oplus {d\mu(r)} W_{r}.\label{U-W}
\eeqa

Finite helicity representations $U_h$ extend to the conformal group (cf. \cite{Mack}). In particular,  by dilation covariance,  we have that $U(\delta(t))U_h U(\delta(t))^*\simeq U_h$  {\bcc hence} $U(\delta(t))U_h|_{\tilde G_3} U(\delta(t))^*\simeq U_h|_{\tilde G_3}$.  Dilations change the unitary class of  $W_{r,2h}$ dilating the radius of the representation\footnote{\pc It is easy to see that the dual action of $\delta(t)$ on a character $\chi_p$ is given by $\delta(t)p=e^tp$.}, since 
\begin{align}
U(\delta(t))U_h|_{\tilde G_3} U(\delta(t))^*= U_h \circ \mrm{ad}(\delta(t))|_{\ti{G}_3}=\int_{\real^+}^{\oplus} d\mu(r) W_{e^{-t} r}.
\end{align}
 Similarly as in   Lemma 4.1 in \cite{LMR16}, this is a consequence of the following computation: 
\begin{align}\label{V-computation-one}
&W_r\circ \mrm{ad}(\delta(t))= W_{r^{-1},2h}  \circ \mrm{ad}(\delta(t))  = \Ind_{\RR^4\rtimes \langle r_1(\pi)\rangle}^{\ti G_3}\left((\chi_{q_1}{\cdot} V_{2h})\circ 
\ad\, { \lambda}_1\left({\ln r}\right) \circ \mrm{ad}(\delta(t)) \right)\nonumber\\
&=\Ind_{\RR^4\rtimes \langle r_1(\pi)\rangle}^{\ti G_3}(\chi_{q_{ e^t r^{-1}}}{\bc \cdot} V_{2h})=W_{e^tr^{-1},2h}=W_{e^{-t}r},
\end{align}
where $q_{ e^t r^{-1}}:=( e^t r^{-1}, e^t r^{-1}, 0,0)$ {\bcc and in the second step we used Lemma 4.1 of \cite{LMR16}.} An analogous computation can be done for $W_0$ {\bcc in order to show that} $W_0\circ{\bcc \mrm{ad}}(\delta(t))\simeq W_0$: since $ {\bcc \mrm{ad}}(\lambda_3)$ does not change the unitary equivalence class of the $\tilde G_3$-representations
\begin{align*}\mathcal{V}_{[r_2(-\pi/2)]}\circ {\bcc \mrm{ad}}(\lambda_3(t))\circ {\bcc \mrm{ad}}(\delta(t))&=\Ind_{H_{r_2(\pi/2)}}^{\tilde G_3}(\chi_{q_0} \cdot \chi_{2h})\circ {\bcc \mrm{ad}}(\lambda_3(t))\circ {\bcc \mrm{ad}}(\delta(t))=\\
&=\Ind_{H_{r_2(\pi/2)}}^{\tilde G_3}(\chi_{q_0} \cdot \chi_{2h})= \mathcal{V}_{[r_2(-\pi/2)]},
\end{align*}
{\bcc where we used} $(\chi_{q_0} \cdot \chi_{2h})\circ {\bcc \mrm{ad}}(\lambda_3(t))\circ  {\bcc \mrm{ad}}(\delta(t))=(\chi_{\lambda_3(-t) \delta(t)q_0} \cdot \chi_{2h})=\chi_{q_0} \cdot \chi_{2h}$ {\bcc referring again to Lemma 4.1 of \cite{LMR16}.}
Therefore, 
\begin{align*}
\int_{\RR^+}^\oplus d\mu(r)W_{r}\simeq   \int_{\real^+}^{\oplus} {d\mu(r)}W_{e^tr}=\int_{\real^+}^{\oplus} {d\mu_t(r)}W_{r},
\end{align*}
where $\mu_t(r){\bc :=}\mu(e^{-t}r)$.  We show that $\mu$ is equivalent to the Lebesgue measure:  assume by contradiction that there exists a set $E\subset \RR^+\backslash\{0\}$ such that $\mu(E)>0$ but $\mu_t(E)=0$ and consider  the multiplication operator by the projection $P_E:={\int^{\oplus}_Ed\mu(r)} \in U_{h}(\tilde G_3)'$. Then  we have that the subrepresentation $P_EU_h|_{\tilde G_3} P_E$ is not contained in 
$U(\delta(t))U_h|_{\tilde G_3} U(\delta(t))^*$ (representations of radius $r\in E$ have measure zero in the latter representation). In particular for every $t\in\RR$, $\mu_t$ is equivalent to $\mu$, hence to the Lebesgue measure on $\RR^+\backslash\{0\}$ up to a possible singular measure in $0$ (see  Proposition~11 of \cite{B04}). Since $\sigma_0$ has null measure in the joint spectrum of translation in $U_h$, by comparing the translation spectrum in left and right side of \eqref{U-W}, $\{r=0\}$ has null $\mu-$measure. Now the statement of the theorem is obtained by a change of variables $r\mapsto \fr{1}{r}$. }\end{proof} 
By {\bcc Propositions \ref{prop:PP},  \ref{Corollary-of-Mackey} and  \ref{prop:P} (ii)} we conclude the modularity condition for finite helicity representations:
\begin{corollary}\label{modularity-finite-helicity}
For every $h\in\frac\ZZ2$, $U_h$  and its multiples satisfy~\eqref{eq:cond}.
\end{corollary}
\begin{proposition}\label{Sum-of-two-helicities}
If $h$ is an integer, namely $U_h$ is bosonic, then  any finite or infinite multiple of $U_h\oplus U_{-h}$ satisfies~\eqref{eq:cond}.
\end{proposition}
\begin{proof}
By Proposition~\ref{Corollary-of-Mackey} $U_h$ and $U_{-h}$ have unitarily equivalent restrictions to ${\tilde G_3} $. 
Indeed, since $h$ is supposed to be integer,  $2h$ is equal to $0$ or $2$ modulo $4$. Clearly {\pc $0\equiv-0$ (mod $4$) and $2\equiv -2$ (mod $4$)}. By Proposition \ref{prop:PP} and Proposition \ref{prop:P}  (ii),(iii), the direct sum $U_h\oplus U_{-h}$ satisfies \eqref{eq:cond}.  Any multiple of $U$ satisfies \eqref{eq:cond} again by Proposition \ref{prop:P} \textit{(ii)}.
\end{proof}
The main result of this section is a corollary of Theorem \ref{thm:starcor}:  
\begin{theoreme}\label{one-particle-net-BW}
Every net of real subspaces  $H$ undergoing the action of  a finite or infinite multiple of $U=U_h\oplus U_{-h}$, {where $h\in \mathbb{Z}$}, satisfies the B-W and the duality properties.
\end{theoreme}

A final remark on finite helicity one particle net{\bcc s} is the following. Massless non-zero finite helicity representations of the Poincar\'e group have to be properly coupled in order to act consistently on a net of  standard subspaces on spacelike cones\footnote{A spacelike cone is a set of the form $C=a+\cup_{\lambda>0}\lambda O$ where  $a\in\RR^{1+3}$ is the apex and $O$ is a double cone which contains only spacelike points and its closure does not contain the origin. It is the intersection of finitely many wedge regions.}. Indeed, by Theorem \ref{thm:starcor},  $U_h$ satisfies the modularity condition \eqref{eq:cond} and any net of standard subspaces it acts covariantly on satisfies the B-W property. Following \cite{GL}, when the B-W property holds then  by spacelike cone localization property one deduces that $U_h$ is covariant under the action of the wedge modular conjugations, namely  $U_h$ extends to an (anti-)unitary representation of the group $\wt\cP_+=\langle\tPoi,\Theta\rangle$, where $\Theta$ is the space and time reflection. The extension is unique up to unitary equivalence 
(Proposition 2.3 {\bcc of} \cite{LMPR} or \cite{NO} for an abstract discussion).
This is not possible for the irreducible finite helicity representations as they are not induced by a selfconjugate representation of the little group (cf. for instance  \cite{var}). In particular any anti-unitary operator implementing the PT symmetry (no charge C considered in this one particle setting) takes $U_h$ into $U_{-h}$.

\section{Bisognano-Wichmann property and asymptotic completeness}\label{last-section}
\setcounter{equation}{0}

In this section we apply  Theorem~\ref{one-particle-net-BW} to a concrete one-particle net 
of standard subspaces in the subspace $\nhil^{(1)}$ from Definition~\ref{single-particle-definition}
and then verify the Bisognano-Wichmann property on $\nhil^{\out}$ using scattering theory.

  Let $\mcF_{\mrm{sa}}(W)\subset \mcF(W)$ be the subspace of self-adjoint operators.  
 It is well known and easy to check that if    $(\mcF,U, \Om)$ is a local net of von Neumann algebras in the
sense of Definition~\ref{HK}, then
\beqa
H(W):={\bc \ov{ \{\, A\Om   \,:\,    A\in \mcF_{\mrm{sa}}(W)\,\}}}, \label{standard-scattering-spaces}
\eeqa
is a net of standard subspaces on wedges w.r.t. $U$, i.e., it satisfies properties 1.-5. of Sect.~\ref{One-particle-nets-sub}.
Motivated by formula~(\ref{single-particle-photon}), we define for any $W\in \mathcal{W}$ the following real subspace of $\nhil^{(1)}$ 
\beqa
H^{(1)}(W):=P^{(1)}H(W). \label{H-one}
\eeqa
Furthermore, we recall that by the Borchers theorem $\De_{\mcF (W)}, J_{\mcF(W)}, S_{\mcF(W)}$
commute with the mass operator and therefore can be restricted to the domain
\beqa 
D:=\{\, P^{(1)}A\Om   \,:\,  A\in \mcF (W)\,\},
\eeqa
which is dense in $\nhil^{(1)}$. (This is proven analogously as Corollary~\ref{Borchers-corollary}). The following result holds:
\begin{proposition}\label{last-section-proposition}
Let $(\A, U, \Omega)$ be a local net of von Neumann algebras  describing massless particles.
\begin{enumerate}
\item[(i)]  The map $H^{(1)}(W)=P^{(1)}H(W)$, defined in (\ref{standard-scattering-spaces})-(\ref{H-one}) above, gives a one particle net of {\bc standard} subspaces on wedges w.r.t.  $U^{(1)}:=P^{(1)}UP^{(1)}$, in the sense of properties 1.-5. of Sect.~\ref{One-particle-nets-sub}. 
\item[(ii)] If the  theory $(\mcF, U,\Om)$ contains massless Wigner particles with helicity zero or with helicities $(h,-h)$, $h\in \nat$, then the one-particle net 
$W\mapsto H^{(1)}(W)$ satisfies  the   B-W and duality properties. 
\item[(iii)]   $J_{H^{(1)}(W)}=J_{\mcF(W)}P^{(1)}$ and $\De_{H^{(1)}(W)}^{it}=\De^{it}_{\mcF(W)} P^{(1)}$ for  $t\in \real$.   
\end{enumerate}
\end{proposition}
\begin{remark} The locality property for the net $W\mapsto H^{(1)}(W)$ can be extracted from \cite{Bu77}, where
\beqa
\lan \Om, A_1P^{(1)}A_2\Om\ran= \lan \Om, A_2P^{(1)}A_1\Om\ran
\eeqa
was obtained for $A_1, A_2$ localized in spacelike separated double cones by the JLD technique. In our context the same property
follows from the Borchers theorem and Lemma~\ref{lem:cons}.
\end{remark}
\begin{proof} (i) By Lemmas \ref{lem:ort1}, \ref{lem:ort2} and \ref{lem:cons}
we have that $H(W)= P^{(1)}H(W)\oplus (1- P^{(1)})H(W)$ and   $W\mapsto P^{(1)}H(W)$   defines a local net of standard subspaces  on $\nhil^{(1)}:=P^{(1)} \nhil$.
It transforms under the massless Poincar\'e representation $P^{(1)}UP^{(1)}$ and satisfies the assumptions 1.-5  in Sect.~\ref{One-particle-nets-sub}. 

(ii) For helicities as in the statement of the proposition we obtain the B-W property for the one-particle net from Theorem~\ref{one-particle-net-BW}. 


(iii) Let $\xi_i\in H^{(1)}(W)$, $i=1,2$, and $ P^{(1)}A_{i,n}\Om$, $n\in \nat$, be the corresponding approximating sequences with $A_{i,n}^*=A_{i,n}$.
Then
\begin{align}
S_{H^{(1)}(W)}(\xi_1+i\xi_2)=\xi_1-i\xi_2&=\lim_{n\to\infty} P^{(1)}  (A_{1,n}\Om -iA_{2,n})\Om     \non\\
&=\lim_{n\to\infty} P^{(1)} S_{\mcF(W)}(A_{1,n}\Om +iA_{2,n})\Om  \non\\
&=\lim_{n\to\infty} S_{\mcF(W)}P^{(1)} (A_{1,n}\Om +iA_{2,n})\Om. 
\end{align}
Hence   $H^{\bc (1)}(W)+iH^{\bc (1)}(W)$ belongs to the domain of the closure of  $S_{\mcF(W)}P^{(1)}$ {\bc and}
the latter operator coincides with $S_{H(W)}$ on $H^{\bc (1)}(W)+iH^{\bc (1)}(W)$. By the uniqueness of the polar decomposition, we
have $J_{H^{(1)}(W)}=J_{\mcF(W)} P^{(1)}$ and $\De^{it}_{ H^{(1)}(W)  } = \De_{\mcF(W)} ^{it}   P^{(1)}$. 
\end{proof}
After this preparation we give a massless version of Lemma 6 and Proposition 7 from \cite{Mu01} and thereby conclude the 
proof of the Bisognano-Wichmann property for asymptotically complete massless theories, as stated in Theorem~\ref{main-theorem}. 
\begin{lemma}\label{lemma-n-particle} In each $n$-particle subspace $\nhil^{(n)}$, $n\geq 2$, there is a total
set of scattering states $\Psi^{\out}:= A_n^{\out}\{f_n\}^+ \ldots A_1^{\out}\{ f_1\}^+\Om$ with the localisation regions $\mco_i$ of $A_i$ chosen s.t.  $\mco_n\subset W_{1}$. 
Furthermore, $(\bv_i-\bv_n)_1<0$ for $\bv_i\in \supp\, f_i$ and $\bv_n\in \supp\, f_n$ for $i\neq n$. 
\end{lemma}
\vspace{-0.2cm}
\begin{proof}  We consider an arbitrary scattering state $\Psi^{\out}=(\Phi_n\tout\cdots\tout \Phi_1)$ constructed using $A_1, \ldots, A_n$ {\bc localized} in some arbitrary double cones 
$\ti\mco_1, \ldots, \ti\mco_n $ and arbitrary {\bc smooth} functions $f_1, \ldots f_n$ on {\bc $S^2$}. By cutting the sphere of velocities into small slices with planes orthogonal to the 1-st axis,
we can approximate $\Psi^{\out}$ with linear combinations of scattering states $\Psi^{\out}_1$ such that the projections of $\supp f_i$ on the 1-st axis are disjoint. 
(Cf.~Proposition~\ref{photon-scattering},  formula~(\ref{single-particle-photon}) and {\bc the} absolute continuity of the momentum spectral measure~\cite{BF82}). Let $f_{i_0}$ be such function that
\beqa 
(\bv_i-\bv_{i_0})_1<0 
\eeqa
for all  $\bv_{i_0}\in \supp\, f_{i_0}$ and $\bv_i\in \supp\, f_i$, $i\neq i_0$. Up to numbering, these scattering states satisfy the condition from the lemma concerning 
velocities, but possibly not the condition concerning localisation regions. Therefore, using again   Proposition~\ref{photon-scattering}, 
formula~(\ref{single-particle-photon})
and the Reeh-Schlieder property for wedges, we approximate each $\Psi^{\out}_1$ by linear combinations of vectors $\Psi^{\out}_2$ constructed using $A'_i$
 localised in double cones $\mco_i$ s.t. $\mco_{i_0}\in W_{1}$.   
Due to the canonical commutation relations of asymptotic fields (cf. formula~(\ref{CCR-asymptotic}) above) each such vector  
 $\Psi^{\out}_2=(\Phi'_n\tout\cdots\tout \Phi'_1)$ coincides with $(\Phi'_{i_0}\tout\cdots\tout \Phi'_n\tout\cdots \tout\Phi'_1)$.  By changing  the numbering, we obtain the claim. \end{proof}
\vspace{-0.6cm} 
\bep\label{final-proposition-scattering} 
 If the unitary groups ${\bc \real\ni s\mapsto} \De^{is}_{\mcF(W_{1})}$ and {\bc $  \real\ni s\mapsto U(\lambda_{W_1}(-2\pi s))$}  coincide on $\nhil^{(1)}$, they also coincide on the subspace $\nhil^{\mathrm{out}}$
of scattering states.
\eep
\begin{proof}Let $V_s:= \De^{is}_{\mcF(W_{1})} U(\Lambda_{W_1}(2\pi s))$. By induction over the particle number $n$, we show that $V_s$ is the unity on each $\nhil^{(n)}$. 
Let $\Psi^{\out}=(\Phi_n\tout\cdots\tout \Phi_1)$ be a scattering state as in Lemma~\ref{lemma-n-particle} with $\Phi_i=A_i^{\out}\{f_i\}\Om$.  
{\bcc For $n=0$ we have $\Phi^{\out}=\Om$ and the statement follows from the invariance of the vacuum under $U(\, \cdot\,)$ and $s\mapsto \De^{is}_{\mcF(W_{1})} $. For $n=1$ the statement holds by Proposition~\ref{last-section-proposition} (ii). Now let $n\geq 2$ be arbitrary and suppose that the statement holds for $n'<n$. } 
Making use of Proposition~\ref{creation-annihilation}~(d), we write
\beqa
\Psi^{\out}\y&=&\y A_n^{\out}\{f_n\}^+\ldots A_1^{\out}\{f_1\}^+\Om\non\\
&=&\y (A_n^{\out}\{f_n\}^++ A_n^{\out}\{f_n\}^-)  \ldots (A_1^{\out}\{f_1\}^+ + A_1^{\out}\{f_1\}^-)\Om+\check\Psi^{\out}\non\\
&=&\y A_n^{\out}\{f_n\}\ldots A_1^{\out}\{f_1\}\Om+\check\Psi^{\out},
\eeqa
where, by the canonical commutation relations for the asymptotic creation/annihilation operators {\bc (cf. formula~(\ref{CCR-asymptotic}))}, the compensating vector $\check\Psi^{\out}$ has components only in 
$\nhil^{(\ell)}$ for $\ell<n$.  Thus we have, by the induction hypothesis, $V_s \check\Psi^{\out}=\check\Psi^{\out}$, and it suffices to consider 
$\hat \Psi^{\out}:=A_n^{\out}\{f_n\}\ldots A_1^{\out}\{f_1\}\Om$. We can write
\beqa
V(s)\hat\Psi^{\out}\y&=&\y V_sA_n^{\out}\{f_n\}\ldots A_1^{\out}\{f_1\}\Om\non\\
\y&=&\y V_sA_n^{\out}\{f_n\} V_s^{-1} V_{s} A_{n-1}^{\out}\{f_{n-1}\}    \ldots A_1^{\out}\{f_1\}\Om\non\\
\y&=&\y V_sA_n^{\out}\{f_n\} V_s^{-1}  A_{n-1}^{\out}\{f_{n-1}\}    \ldots A_1^{\out}\{f_1\}\Om,
\eeqa
where in the last step we used the induction hypothesis. (We also used that the state $A_{n-1}^{\out}\{f_{n-1}\}    \ldots A_1^{\out}\{f_1\}\Om$ is in $D_{P^0}$,
cf. Proposition~\ref{creation-annihilation}, and the fact that due to  the Borchers theorem, $V_s$ leaves $D_{P^0}$ invariant. Consequently, the limit $A_n^{\out}\{f_n\}$ still exists in the last line above).
Using again that $V_s$ commutes with translations,  we have $V_{s}A_n^{\out}\{f_n\} V_{s}^{-1}= (A_n')^{\out}\{f_n\}$, where $A_n':=V_sA_nV_s^{-1}$ is
localized in $W_{1}$. Thus we can write
\beqa
V(s)\hat\Psi^{\out}\y&=&\y (A_n')^{\out}\{f_n\}  A_{n-1}^{\out}\{f_{n-1}\}   \ldots A_1^{\out}\{f_1\}\Om\non\\
\y&=&\y  \sum_{i=n-1}^1 A_{n-1}^{\out}\{f_{n-1}\}    \ldots [(A_n')^{\out}\{f_n\},  A_i^{\out}\{f_i\}]    \ldots A_1^{\out}\{f_1\}\Om\quad \label{terms-to-zero} \\
\y &+ & \y A_{n-1}^{\out}\{f_n\}\ldots A_1^{\out}\{f_1\}  (A_n')^{\out}\{f_n\}\Om. \label{last-term-comm}
\eeqa
Concerning (\ref{last-term-comm}), we use that $(A_n')^{\out}\{f_n\}\Om=V_{s}A_n^{\out}\{f_n\} V_{s}^{-1}\Om=A_n^{\out}\{f_n\}\Om$,
since $V_s$ preserves both the vacuum and single-particle subspace.  This vector coincides with $\hat\Psi^{\out}$ provided that
we can show 
\beqa
[A_n^{\out}\{f_n\}, A_i^{\out}\{f_i\}]=0
\eeqa
for all $i=1,2,\ldots, n-1$. This follows from  formulas~(\ref{single-particle-photon}), (\ref{c-number-commutator}) and the disjointness
of supports of $f_n, f_i$ (cf. Lemma~\ref{lemma-n-particle}).    
Thus to conclude the proof of the proposition we have to show that the terms in (\ref{terms-to-zero}) are zero. Since $A'_n$ is only wedge-local, we
 cannot use  property~(\ref{c-number-commutator}) and we need to proceed via a direct computation: For unit vectors $\Psi_1,\Psi_2\in D_{P^0}$ we write
\beqa
 & &|\lan \Psi_1 [(A_n')^{\out}\{f_n\},  A_i^{\out}\{f_i\}]\Psi_2\ran|\non\\
& &\leq  \ov{\lim}_{t\to\infty} \int dt'_n dt'_i d\om(\bn_n)d\om(\bn_i) \,  h_t(t'_n)h_t(t'_i) |f_n(\bn_n)| |f_i(\bn_i) 4t'_nt'_i\times \non\\
& &\ph{4444444444444444}\times \|[\pa_0A'_n(t'_n, t'_n\bn_n), \pa_0A_i(t'_i, t'_i\bn_i)]\| \non\\
& &\leq  \ov{\lim}_{t\to\infty} \int dt'_n dt'_i d\om(\bn_n)d\om(\bn_i) \,  h_t(t'_n)h_t(t'_i) |f_n(\bn_n)| |f_i(\bn_i) 4t'_nt'_i\times \non\\
& &\ph{4444444444444444}\times \|[\pa_0A'_n, \pa_0A_i(t'_i-t'_n, t'_i\bn_i -t'_n\bn_n  )]\| \non\\
& &\leq  \ov{\lim}_{t\to\infty} \int dt''_n dt''_i d\om(\bn_n)d\om(\bn_i) \, {\bcc t^{2\eps}} h(t''_n)h(t''_i) |f_n(\bn_n)| |f_i(\bn_i) |4(t+t^{\eps}t''_i) (t+t^{\eps}t''_n) \times \non\\
& &\ph{4444444444444444}\times \|[\pa_0A'_n, \pa_0A_i(t^{\eps} (t''_i-t''_n), (t+t^{\eps}t''_i)\bn_i -(t+t^{\eps}t''_n)\bn_n  )]\|. \quad\quad
 \eeqa
Making use of Lemma~\ref{lemma-n-particle}, we conclude that the last expression is zero for sufficiently large $t$. Indeed, $t( \bn_i-\bn_n)_1<0$,
$0<\eps<1$ and $t_i'', t_n''$ are restricted to unit balls around zero. Hence $\pa_0A_i(t^{\eps} (t''_i-t''_n), (t+t^{\eps}t''_i)\bn_i -(t+t^{\eps}t''_n)\bn_n  )$
is eventually localized in the left wedge. \end{proof}
 \section{Conclusion and outlook}
\setcounter{equation}{0}

In this paper we proved the Bisognano-Wichmann property for asymptotically complete theories of massless particles
with integer helicities.  The  argument starts from verifying this property at the single particle level. 
For this purpose, the single-particle subspace $\nhil^{(1)}$ is equipped with the structure of a local net of standard subspaces.
This net is covariant w.r.t.  a representation $U^{(1)}:=U |_{ \nhil^{(1)} }$ of the Poincar\'e group, which is a direct sum of two representations of opposite integer helicities, i.e., $U^{(1)}=U_{h}\oplus U_{-h}$ {\bc (or a multiple thereof)}.
 Then we verified the modularity condition for the B-W property $U^{(1)}(r_W)\in U^{(1)}(G_W)''$, where $G_W$ is  the subgroup of Poincar\'e transformations 
preserving a wedge $W$ and $r_W$ maps $W$ to the opposite wedge. This technically demanding step
was accomplished by showing that $U_h$ and $U_{-h}$ have the same restriction to the group $\ti G_W$ generated by $G_W$ and $r_W$. 
Hence $U^{(1)}{\bc |_{\ti G_W}}= {\bc U_h|_{\ti G_W}}\otimes 1$ and the modularity condition could be concluded from earlier results \cite{Mor}. Given the
B-W property at the single-particle level, the B-W property of the full theory was verified using scattering theory and asymptotic completeness.
In this part we adapted the arguments of Mund \cite{Mu01} to the massless case.

A natural question for future research is a generalization of our arguments to   particles with half-integer helicities.  The obstruction comes from the fact that in this case $U_{h}$ and $U_{-h}$  are unitarily equivalent when restricted to $G_3$ but have disjoint restrictions to $\tilde G_3$
 and one cannot apply Proposition \ref{prop:P} (iii).
Another  future research direction is to  relax the assumption of asymptotic completeness. We remark that in the vacuum
sector of QED asymptotic completeness of photons can be assumed only below a certain energy threshold, excluding the electron-positron
pair production. It is an interesting question how to prove the B-W property in this physically relevant situation. 
In this context we remark  that  our results give the B-W property of the net of asymptotic photon fields of QED, defined at the end of Sect.~\ref{massless-section}. This net plays an important
role in the study of infrared problems  (see e.g. \cite{Bu77, BD84, AD17}) and we hope that our results will  also find applications
there.

\appendix

\section{Direct integral representations}
\setcounter{equation}{0}

We suggest \cite{Tak, Dix, MT} as further references for basic definitions.

Given a field of Hilbert spaces $\gamma\mapsto \H(\gamma)$ on a standard measure space $(\Gamma,\mu)$, the direct integral Hilbert space $\int_\Gamma^\oplus \H(\gamma) d\mu(\gamma)$ is defined if the field is $\mu$-measurable.
 This definition requires and depends on the choice of a linear $\gamma$-pointwise dense subspace $\mathcal{S}$  of the topological product $\Pi_{\gamma\in \Gamma}\H(\gamma)$ 
   which selects a family of $\mu$-measurable vector fields. 
 (cf.\! \cite[Part II, Sect. II.1.3, Definition 1]{Dix}). Note that given a sequence of measurable vector field $\xi_n$ $\mu$-a.e.\ pointwise converging to $\xi$, namely $\|(\xi_n)(\gamma)-(\xi)(\gamma)\|_{\bc \ga} \rightarrow0$ for $\mu$-a.e.\ $\gamma\in \Gamma$, we obtain that $\xi$ is a $\mu$-measurable vector field. We also recall that a vector field of bounded operators $ \ga\mapsto T(\gamma)\in \B(\H(\ga))$ is $\mu$-measurable if for any $\mu$-measurable field $\ga\mapsto \xi(\gamma)\in\H(\gamma)$ we have that $\ga\mapsto T(\gamma)\xi(\gamma)\in \H(\gamma)$ is $\mu$-measurable.  In this case we define
 $$T\xi:=\int^\oplus_\Gamma T(\gamma)\xi(\gamma)d\mu(\gamma)\in\int_\Gamma^\oplus \H(\gamma)d\mu(\gamma).$$
 We write this operator as $T=\int_\Gamma^\oplus T(\gamma) d\mu(\gamma)$ and $T$ is called the direct integral of ${\bc \ga\mapsto} T(\gamma).$  Furthermore, we have that $\gamma\mapsto\|T(\gamma)\|_{\bc \ga}$ is measurable and $\|T\|=\sup_{\gamma \in \Gamma}\|T(\gamma)\|_\gamma$. The operators of this form are said to be \textit{decomposable}. If $T(\gamma)$ is a scalar for any $\gamma{\bc \in \Ga}$, then $T$ is said to be a \textit{diagonal }operator. The algebra generated by the diagonal operators is called the \textit{diagonal }algebra. Note that any operator $T\in\B(\H)$ is decomposable iff $T$ commutes with the diagonal algebra (cf.  \cite{Tak} IV.8, 
Corollary~8.16).
 A field of von Neumann algebras $\gamma \mapsto \M(\gamma)\subset  \B(\H(\gamma))$ is said to be measurable if there exists a countable family $\{x_n(\gamma)\}_{n\in\NN}$ of measurable fields of operators s.t. $\M(\gamma)$ is generated by $x_n(\gamma)$ for (a.e.) $\gamma\in\Gamma$. Then it is possible to define $\M=\int_\Gamma^\oplus \M(\gamma) d\mu(\gamma)$. Given a $C^*$-algebra $\A$, a field of continuous representations $\gamma\mapsto\pi(\gamma)$ is said to be measurable if for any $A\in\A$, the operator field $\gamma\mapsto \pi(\gamma)(A)$ is measurable. Then one can define the representation $\int_\Gamma^\oplus \pi(\gamma){d\mu(\gamma)}$ acting on $\int_\Gamma^\oplus\H(\gamma)d\mu(\gamma)$. 

\section{Proof of Proposition \ref{prop:PP} }{\label{proposition}  }
\setcounter{equation}{0}

For fundamental concepts on direct integral of representation see  Appendix A. 

The function of the translation generators $P_1^2+P_2^2$ is a Casimir operator for $G_3$ and decomposes {\bc according} to $\mu$, i.e., $P_1^2+P_2^2=\int^{\oplus}_{\real{^+}}r^2\cdot 1 d\mu(r)$ and $P_1^2+P_2^2$ is affiliated to $U(G_3)''$.
By definition, bounded functions of $r={\pc\sqrt{P_1^2+P_2^2}}$ generate the diagonal algebra 
$\mathcal{D}$. Thus $\mathcal{D}$ is contained {\pc in the center of} $U(G_3)''$, hence any operator in $U(G_3)'$ is decomposable since  it commutes with $\mathcal{D}$.

We now pick $T\in U(G_3)'$  then $T$ is a decomposable operator, namely, 
\beqa
T =\int_{\real^+}^\oplus \,  T(r) d\mu(r).
\eeqa

Assume that there exists a positive measure set $I$ s.t. $T(r)$  is not in $U_r(G_3)'$ for $r\in I$. 
Let $\chi$ be a characteristic function of $I$.  Then $ \int_{\real^+}^\oplus \, T(r) \chi(r)d\mu(r)$ is not in the commutant of $U(G_3)$, which is a contradiction. We conclude that
\beqa
U(G_3)'=\int_{\real{^+}}^\oplus U_r(G_3)' d\mu(r)
 \eeqa
and thus 
\beqa \label{eq:dis} 
U(G_3)''=\int_{{\real^+}}^\oplus  U_r(G_3)'' d\mu(r)
\eeqa 
by  \cite[Theorem 8.18]{Tak}. 
{\bc Now we recall that   $U_r$ satisfies \eqref{eq:cond} by Corollary \ref{Mod-Corollary}. Since, by assumption, $U(r_1(\pi))=\int_{{\real^+}}^\oplus  U_r(r_1(\pi)) d\mu(r)$, we obtain from \eqref{eq:dis} that $U$ satisfies \eqref{eq:cond}.}

\section{Regularity of the actions  on $\RR^{{\bc 1+3}}$}
 \label{App-C}
\setcounter{equation}{0}

\begin{definition}\label{regular-definition} Let $G$ be a locally compact, $\sigma$-compact, group and $N$ be a normal abelian subgroup, then the (dual-)action of $G$ on $\hat N$ is regular if
\begin{enumerate}
\item[R1.] the orbit space is countably separated, namely there exists a countable family $\{E_n\}_{n\in\NN}$ of $G$-invariant Borel sets in $\hat N$ s.t.  each orbit in $\hat N$ is the intersection of all $E_n$ that contain it,
\item[R2.] each orbit is relatively open in its closure.
\end{enumerate}
\end{definition}
{\pc Here we check that the action of $G_3^0$  and $\tilde G_3^0$ on $\RR^{{\bc 1+3}}$ is regular according to the previous definition.}

R1.
Let $\textbf{o}$ be a $G_3^0$ or a $\tilde G_3^0$ orbit on $\RR^{1+3}$. Then $\textbf{o}$ can be obtained by intersection of the subsets of the following \textit{countable} family of Borel subsets containing $\textbf{o}$. {\bc For} $a_1,a_2,b_1,b_2\in\mathbb Q$, $c,d\in \mathbb Q^{\geq0}$, consider the following sets:
\begin{itemize}
\item $A_{a_1,b_1}=\{p: a_1\leq p^2\leq b_1\}$ and $A_{a_1,b_1}^\pm=A_{a_1,b_1}\cap \{p:\pm p_0>0\}$ if $a_1,b_1\geq0$,
\item $E_{c,d}=\{p:c\leq p_1^2+p_2^2\leq d\}$,
\item $F_{a_2,b_2}=\{p:a_2\leq p_0^2-p_3^2\leq b_2\}$,
\item $F_{a_2,b_2}^\pm=\{p:a_2\leq p_0^2-p_3^2\leq b_2, \pm p_3>0\}$ if $a_2, b_2<0$,
\item $K^{\pm ,\pm}=\{p=(p_0,0,0,p_3): p_0=\pm p_3, \pm p_0>0\}$,
\item $\tilde K^{\pm}=\{p=(p_0,0,0,p_3): p_0=\pm {\bc |p_3|}, \pm p_0>0\}$,
\end{itemize}
{\bc where in the case of $K^{\pm ,\pm}$ the two signs are uncorrelated.}
We shall denote with $\A, \E, \F, \F^{\pm},\K^{\pm\pm}, \tilde \K^\pm$  the countable families of the above sets with the corresponding letters. We also define the sets:

\begin{itemize}

\item $O=\{p:p=0\}$,
\item $Z_{a_1,b_1,\pm,\pm}=\{p:a_1\leq p^2\leq b_1<0, p_0^2-p_3^2=0, \pm p_3>0, \pm p_0> 0\}$, 
\item $\tilde Z_{a_1,b_1,\pm}=\{p:a_1\leq p^2\leq b_1<0, p_0^2-p_3^2=0, \pm p_0> 0\}$, 
\item $X_{a_1,b_1}=\{p:a_1\leq p^2\leq b_1<0, p_0=p_3=0\}$.
\end{itemize}

In the following we shall say that  a family of sets {\it selects }an orbit $\textbf{o}$ if  the latter is the intersection of all the set of the family containing $\textbf{o}$. The set selecting an orbit of a group will be invariant  under the group action. All the {\bc families} we will consider will be countable as well as their union.

Firstly,  any {\bc orbit of} $G_3^0$ or $\tilde G_3^0$ is contained in a Lorentz orbit in $\RR^{1+3}$. The family in $\A$ selects the Lorentz orbits. The orbit in the origin is selected by $O$.  Now  $G_3^0$ and $\tilde G_3^0$ share the same {\bc massive orbits, contained in $p^2=m^2$, $m>0$,} that can be selected by considering $\A$ and $\E$ families. 
Now consider a massless orbit in the forward lightcone. If for every $p\in \textbf{o}$, $p_1^2+p_2^2>0$ then it is  both a $G_3^0$ and $\tilde G_3^0$ orbit {\bc (cf. Remark~\ref{rmk:orb})} and can be selected by the families $\A$ and $\E$. If  $p_1^2+p_2^2=0$ then the two $G_3^0$ orbits $\{p:0<p_0=p_3\}$ and
$ \{p:0<p_0=-p_3\}$ are selected by ${\bc  K}^{\pm,{\bc +}}$.  If  $p_1^2+p_2^2=0$ then the $\tilde G_3^0$ orbit is selected by ${\bc \tilde K^+}$. 
{\bc We argue analogously for the backward lightcone, referring to sets $K^{\pm,-}, \ti K^-$}.
Now consider imaginary mass {\bc orbits, contained in} $p^2=-m^2$ and assume that ${\bc p_1^2+p_2^2=r^2}$. We have three cases:
\begin{itemize}
\item $r^2<-m^2$. In this case we have two branches of the hyperboloid $p_0^2-p_3^2=m^2+r^2<0$ that become two $G_3^0$ orbits and a unique  $\tilde G_3^0$ orbit. The $G_3^0$ and $\tilde G_3^0$ {\bc orbits} are  selected by $\A,\E,\F^\pm$ and $\A,\E,\F$, respectively.
\item $r^2=-m^2$. $G_3^0$ orbits are selected by $Z_{a_1,a_2,\pm,\pm}$ or  $X_{a_1,a_2}$. $\tilde G_3^0$ orbits are selected by $\tilde Z_{a_1,a_2,\pm}$.
\item $r^2>-m^2$. $G_3$ and $\tilde G_3$ share the same orbits selected by $\A, \E, \F\cap U^\pm$, where $U^\pm= \{p:\pm p_0\geq0\}$.
\end{itemize}

R2. trivially holds.

\end{document}